\pdfoutput=1

\documentclass[11pt]{amsart}
\usepackage{hyperref}
\usepackage{dsfont}
\usepackage[all]{xy}
\usepackage{graphics}
\usepackage{epsfig}
\usepackage{amsmath}
\usepackage{amscd}
\usepackage{pdfsync}
\usepackage{comment}
\usepackage{todonotes}
\usepackage{mdwlist}
\allowdisplaybreaks
\date\today

\usepackage{float}
 
\usepackage[shortlabels]{enumitem}
\usepackage{epsfig}
\usepackage{amsmath}
\usepackage{amscd}
\usepackage{tikz-cd}
\usetikzlibrary{decorations.pathmorphing}
\usepackage{verbatim}
\usepackage{pdfpages}
\usepackage{amsfonts,latexsym,amssymb}
\usepackage{parskip}
\usepackage{cancel}

\usepackage{pgfplots} 
\usepgfplotslibrary{statistics}

\usepackage{tikz,tikz-3dplot}

\usepackage{amsfonts,latexsym,amssymb}

\newcommand{\bbL}{{{\mathbb{L}}}}
\newcommand{\bbE}{{{\mathbb{E}}}}
\newcommand{\bbZ}{{{\mathbb{Z}}}}

\newcommand{\bbR}{{{\mathbb{R}}}}

\newcommand{\bbV}{{{\mathbb{V}}}}

\newcommand{\cA}{{\mathcal{A}}}

\newcommand{\cD}{{\mathcal{D}}}

\newcommand{\cG}{{\mathcal{G}}}

\newcommand{\cJ}{{\mathcal{J}}}

\newcommand{\cL}{{\mathcal{L}}}
\newcommand{\cM}{{\mathcal{M}}}

\newcommand{\cP}{{\mathcal{P}}}

\newcommand{\cT}{{\mathcal{T}}}

\newcommand{\be}{\mathbf e}

\newcommand{\bt}{\mathbf t}  
\newcommand{\bu}{\mathbf u}

\newcommand{\bA}{\mathbf A}

\newcommand{\bG}{\mathbf G}

\newcommand{\bM}{\mathbf M}

\newcommand{\bS}{\mathbf S}

\newcommand{\tC}{{\tilde{C}}}
\newcommand{\tP}{{\tilde{P}}}

\newcommand{\tZ}{{\tilde{Z}}}

\newcommand{\tLambda}{{\tilde{\Lambda}}}
\newcommand{\tdt}{{\widetilde{dt}}}
\newcommand{\tGamma}{{\tilde{\Gamma}}}

\newcommand{\vu}{\vec{u}}

\newcommand{\hchi}{{\hat{\chi}}}
\newcommand{\hphi}{{\hat{\phi}}}

\newcommand{\teta}{{\tilde{\eta}}}

\newcommand{\tgamma}{{\tilde{\gamma}}}
\newcommand{\tsum}{{\widetilde{\sum}}}

\newcommand{\sfv}{{\mathsf{v}}}
\newcommand{\sfb}{{\mathsf{b}}}
\newcommand{\sfe}{{\mathsf{e}}}

\newcommand{\sfP}{{\mathsf{P}}}
\newcommand{\tsfP}{{\tilde{\mathsf{P}}}}

\newtheorem{proposition}{Proposition}[section]

\newtheorem{lemma}[proposition]{Lemma}

\newtheorem{corollary}[proposition]{Corollary}

\theoremstyle{definition}

\newtheorem{remark}[proposition]{Remark}

\begin{document}

\title[Character varieties and $SU(2)$ Lattice Gauge Theory I]{Symplectic Geometry of character varieties and $SU(2)$ Lattice Gauge Theory I}

\author[T. R. Ramadas]{T. R. Ramadas}
\address{Chennai Mathematical Institute\\
H1 Sipcot IT Park\\  
Siruseri, TN 603103\\
India} 
\email{ramadas@cmi.ac.in}

\date{\today}

\thanks{I thank Mahan Mj for tutorials on pants decompositions and K.N. Raghavan for help with Schur-Weyl duality.  Thanks also to S. Gupta, A. Ladha and  P.K. Mitter for valuable feedback regarding this manuscript. This work is partially supported by the Infosys Foundation and (in its initial stages) by the Department of Science and Technology, {\it via} a J.C. Bose Fellowship.}

\begin{abstract} Given a finite connected graph $\Lambda$, the space of $SU(2)$ lattice gauge-fields on $\Lambda$, modulo gauge transformations, is a Lagrangian submanifold -- with mild singularities -- of the $SU(2)$ character variety (= phase-space of Chern-Simons theory) of an associated surface. We present evidence that, in the limit of large $\Lambda$, integration over the character variety with respect to the Liouville measure approximates lattice-theoretic integrals. By the works of W. Goldman, L. Jeffrey and J. Weitsman, the formalism of Duistermaat-Heckman applies  to the relevant integrals over the character variety. A continuous version of the Verlinde algebra facilitates computations. In two dimensions we recover standard expressions. For the theory on a 3-dimensional periodic lattice $\Lambda$ with Migdal action we get a very pleasant expression for the symplectic partition function, and with the Wilson action a more elaborate one. Each is a sum of a series with positive terms. One can also write down expressions for plaquette-plaquette correlations and 't Hooft loops.
\end{abstract}

\maketitle

\pagebreak

\section{Introduction and Summary}

In the diagrams below we use a simple example to visualise the definitions; we deal with more interesting cases later. 

\begin{enumerate}[wide, labelwidth=!, labelindent=0pt,label=(\alph*)]

\item Let $\Lambda$ be a finite connected graph (which we will refer to as a lattice), with $\sfv$ vertices and $\sfe$ edges. Its first Betti number is $\sfb = \sfe-\sfv+1$. We will suppose that $\sfb \ge 2$.

\begin{figure}[H]
\includegraphics[scale=0.5,trim=20 20 20 20]{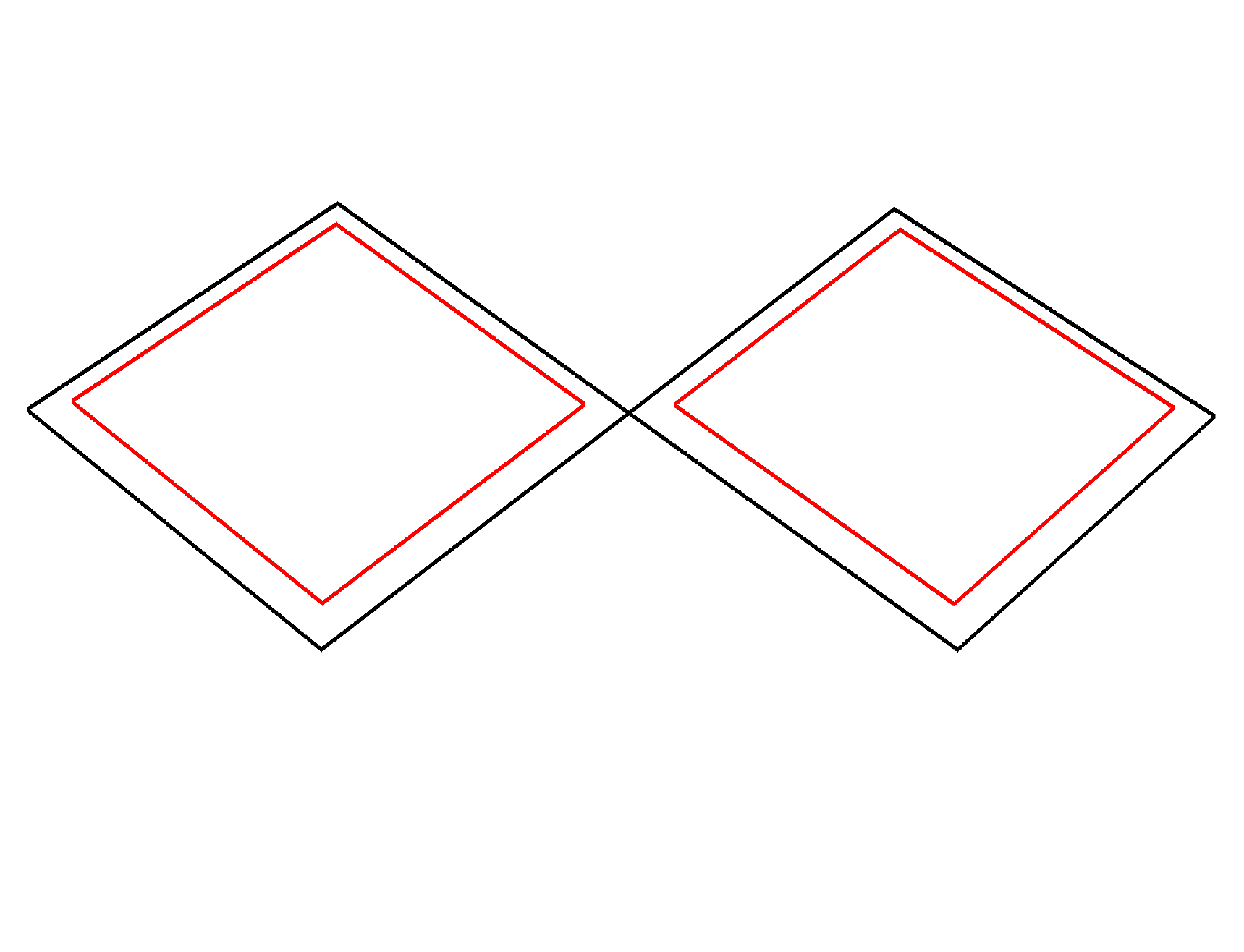}
\caption{A connected graph with b=2. $\sfP$ is the set of two red loops, for clarity drawn slightly displaced from the graph itself.}\label{GraphForSummary}
\end{figure}

We let $\bA$ denote the space of $SU(2)$ lattice gauge fields on $\Lambda$, and $\bG$ the group of lattice gauge transformations\footnote{These terms are explained in \S \ref{lattice}.}. The central objects are integrals of the form
\begin{equation}\label{latticeintegral}
\int_\bA \exp{\{\sum_{\gamma \in \sfP} \beta_\gamma \phi_\gamma (g)\}} \ F(g) \cD g
\end{equation}
Here $\cD g$ is the Haar measure on $\bA$ normalised so that $\int \cD g=1$,  $\sfP$ is a set of closed paths  in $\Lambda$, $\beta_\gamma$ are real numbers, $\phi$ is a class function on $SU(2)$, $\phi_\gamma(g)=\phi(H_\gamma(g))$ where $H_\gamma(g)$ is the lattice holonomy (of the lattice gauge field configuration $g \in \bA$) around $\gamma$, and lastly, $F$ is a gauge-invariant function. (Note that for fixed $\phi$ and $\gamma$, $\phi_\gamma$ is one such.)

We let $\bM$ denote the quotient $\bA/\bG$.

\item \emph{``Pumping up": graphs and handlebodies.}  One can associate to $\Lambda$ a handlebody $M_\Lambda$ as follows: embed $\Lambda$ in a(ny) three-manifold $M^3$ and take $M_\Lambda$ to be a closed regular neighbourhood. (Informally, a 3-dimensional handlebody is a 3-manifold -- with boundary -- which is diffeomorphic to a closed 3-ball with a finite number of three-dimensional cylinders attached to it.) The graph $\Lambda$ sits naturally inside the handlebody $M_\Lambda$ as a \emph{spine}. The boundary of the handlebody is a closed orientable surface $\bS = \bS_\Lambda$ without boundary.  There is a retraction $r:M_\Lambda \to \Lambda$.

\begin{remark}\label{Heegaard} The three-manifold $M^3$ could be $\bbR^3$. If, on the other hand, $M^3$ is compact, the closure $M'_\Lambda$ of the complement of $M_\Lambda$ is a compact 3-manifold sharing $\bS$ as a common boundary with $M_\Lambda$. (If $M'_\Lambda$ is also a handlebody, we have a \emph{Heegaard splitting} of $M^3$.)
\end{remark}

\begin{figure}[H]
\includegraphics[scale=0.5,trim=20 20 20 20,clip]{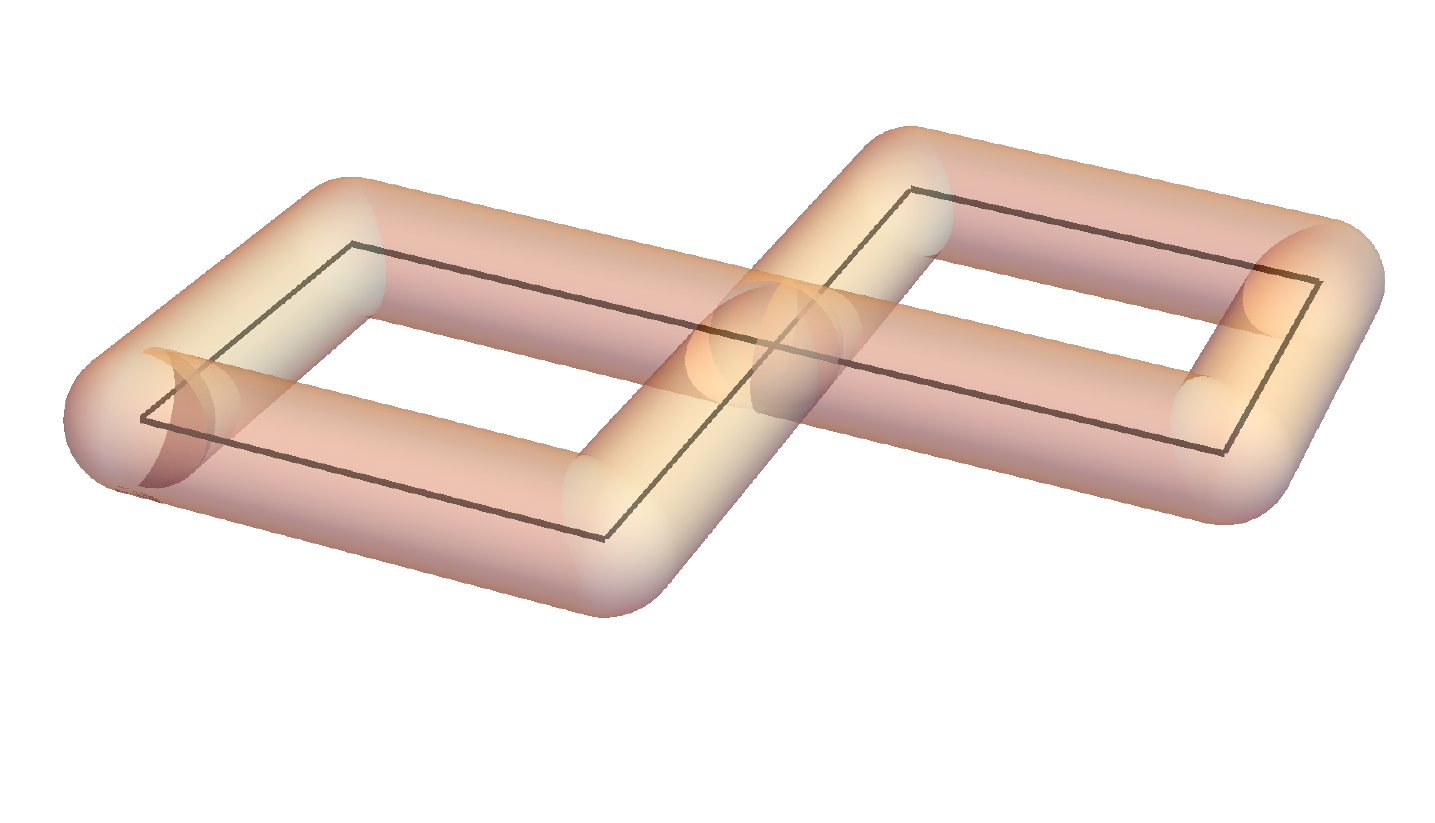}
\caption{The handlebody, spine (in black) and boundary surface (of genus 2).}\label{LatticeOfTubesForSummary}
\end{figure}

The first Betti number of the surface $\bS$ is $2(\sfe-\sfv+1)$.  (This follows, for example, by a simple Mayer-Vietoris argument.)
The genus\footnote{This is \emph{not} the genus of the graph $\Lambda$.} of $\bS_\Lambda$ is therefore $\sfb = \sfe-\sfv+1$.

\item Let $\cM \equiv \cA_f/\cG$ where $\cA_f$ is the space of flat $SU(2)$ gauge-fields (in differential-geometric terms, flat connections on the trivial $SU(2)$ bundle) on $\bS$, and $\cG$ is the group of gauge-transformations. This is a manifold -- outside some mild singularities, which we ignore  -- of dimension $6(\sfb-1)$.

The manifold $\cM$ carries a symplectic structure outside its singularities and the corresponding Liouville measure $d \mu_\cL$. We will assume that the symplectic structure is normalised so that the total volume is one. (The symplectic structure\footnote{Up to normalisation, this is the symplectic structure that $\cM$ inherits as the phase-space of $SU(2)$ Chern-Simons theory.} depends on a choice of orientation on $\bS$, which we will assume made. The measure-theoretic statements below are independent of this choice.)

The space $\bM$ can be canonically identified with the moduli space of flat connections on $\bS$ that extend (as a flat connection) to the interior of $M_\Lambda$. In fact, $\bM \subset \cM$ is a Lagrangian submanifold. In particular the dimension of $\bM$ (which, like $\cM$, has mild singularities) is $3(\sfe-\sfv)=3(\sfb-1)$.

\item The space $\cM$ is the space of isomorphism classes of representations of $\pi_1(\bS)$ in $SU(2)$, and thus a real affine algebraic variety -- in this avatar known as a \emph{character variety} -- and $\bM \hookrightarrow \cM$ a subvariety. A natural class of algebraic functions on $\cM$ (respectively, $\bM$) consists of the \emph{Wilson loop functions} $W_\tgamma$ (respectively, $W_\gamma$), namely, traces of the holonomy around closed  loops $\tgamma$ in $\bS$ (respectively, $\gamma$ in $\Lambda$). These generate the algebra of all algebraic functions \cite{CharlesMarche}. Note that if $\tgamma$ is a loop in $\bS$ and $\gamma=r \circ \tgamma$ its image under the retraction $r:\bS \to \Lambda$,
\begin{equation}\label{restrict}
W_\tgamma|_\bM=W_\gamma
\end{equation}

\item Since $\bM$ is a Lagrangian submanifold of $\cM$, it does not inherit a measure naturally from $d\mu_\cL$. We shall present evidence that nonetheless, 

$\#$ \textbf{for large graphs $\Lambda$, integrals over $\bM$ w.r.to the Haar measure $\cD g$ of (products of functions) $W_\gamma$  can be approximated by integrals of the corresponding   (products of functions) $W_\tgamma$ over $\cM$ w.r.to the Liouville measure $d\mu_\cL$, normalised so that the total volume of $\cM$ is 1.}

\item\label{DH} The advantage of working with the measure $d\mu_\cL$  is that we can appeal to the  following circle of results, due primarily to W. Goldman, L. Jeffrey and J. Weitsman (\cite{Goldman},\cite{JeffreyWeitsman}):

\emph{Consider a pants decomposition of $\bS$ with circles $\tGamma \equiv \{\tgamma_1,\dots,\tgamma_{3(\sfb-1)}\}$ bounding the pants. Let $W_{\tgamma_i}$ denote the corresponding Wilson loop functions, and define functions $t_i:\cM \to [0,1)$ by $W_{\tgamma_i}=2\cos(\pi t_i)$. Let $\bt:\cM \to [0,1]^{3(\sfb-1)}$ be the map with components $t_i$. Let $\cT_\tGamma$ be the corresponding trinion graph. We shall suppose that $\cT_\tGamma$ has no loops and cannot be disconnected by removing one edge. Then the closure of the image of $\cM$ by $\bt$ is the polytope $\cP_{\cT_\tGamma}$ defined by the conditions:
\begin{itemize}
\item For each pair of pants in the decomposition, (i.e., each vertex of ${\cT_\tGamma}$) the components $(t_{i_1},t_{i_2},t_{i_3})$ of $\bt$ corresponding to the three boundary circles (i.e., the edges meeting at the vertex) obey the spherical triangle inequalities (see below).
\end{itemize}
Further, the push-forward of the Liouville measure by $\bt$ is the Lebesgue measure times the indicator function of the polytope $\cP_{\cT_\tGamma}$, normalised so that the its integral is 1.}

(This is closely related to the Duistermaat-Heckman Theorem.)

\item Returning to the lattice theory, suppose given a family $\sfP$ of loops in the graph $\Lambda$, for example as in the integral (\ref{latticeintegral}). Suppose there exists a judiciously chosen set $\tGamma$ of $3(\sfb-1)$ closed loops in $\bS$ that give a pants decomposition, and a subset $\tsfP$ of $\tGamma$ such that
\[
\tgamma_i \mapsto \gamma_i=r \circ \tgamma_i
\]
is a bijection $\tsfP \to \sfP$. Then we can consider the map $\bt_\bM:\bM \to [0,1]^{3(\sfb-1)}$ defined as above (but using the functions $W_{\gamma_i}$ rather than $W_{\tgamma_i}$). Clearly the closure of the image of $\bt_\bM$ is contained in the polytope $\cP_{\cT_\tGamma} \subset [0,1]^\tGamma$. (In the most relevant cases, the image coincides with the polytope.)

In our running example, this is illustrated here:

\begin{figure}[H]
\includegraphics[scale=0.7,trim=20 80 20 20,clip]{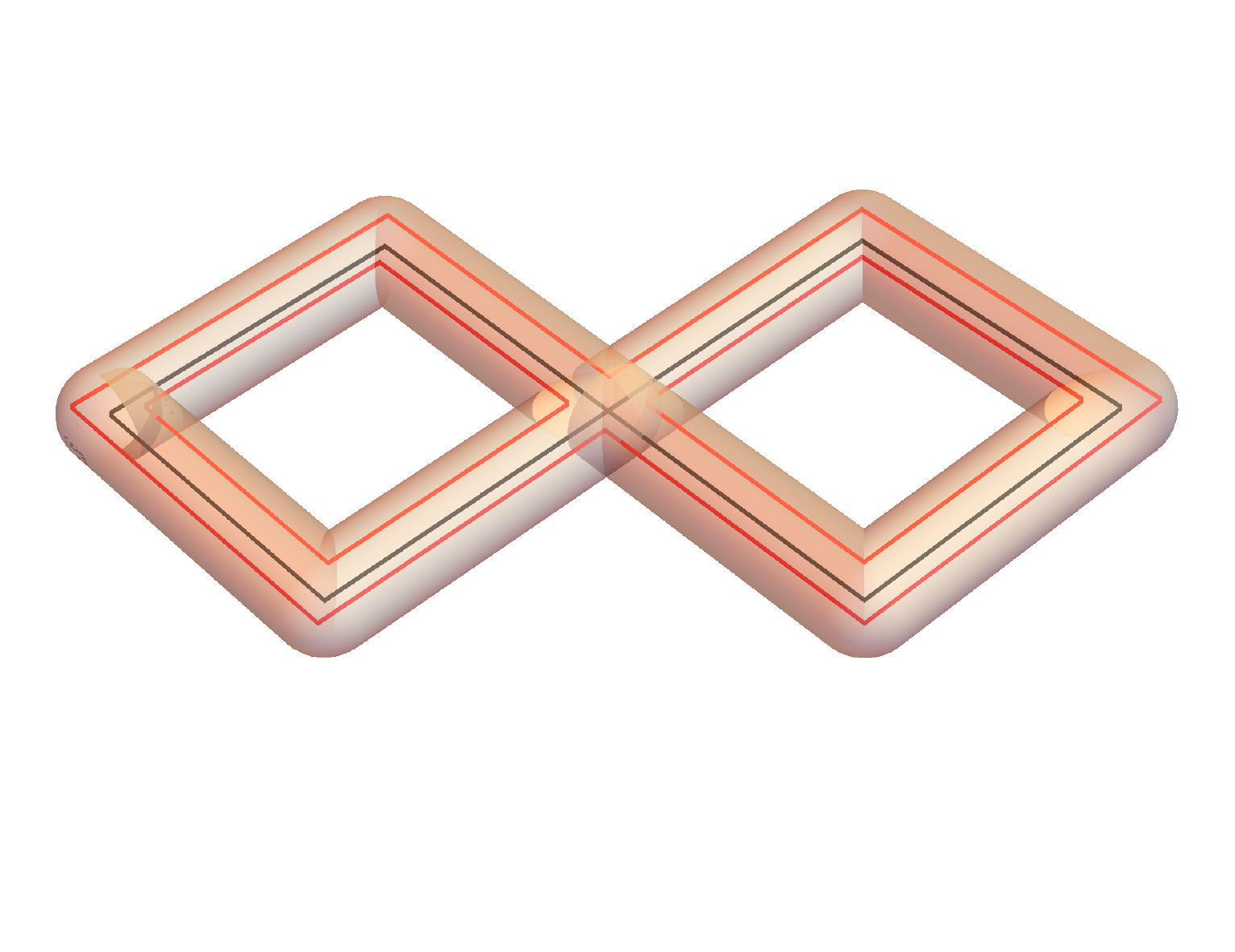}
\caption{Three closed loops (red) forming a pants decomposition $\tGamma$ of $\bS_\Lambda$. $\tsfP$ is the set of two smaller red loops. The spine is shown in blue.}\label{Tubes3DPlaquettesForSummary}
\end{figure}

...and here:

\begin{figure}[H]
\includegraphics[scale=0.5,trim=25 25 25 25]{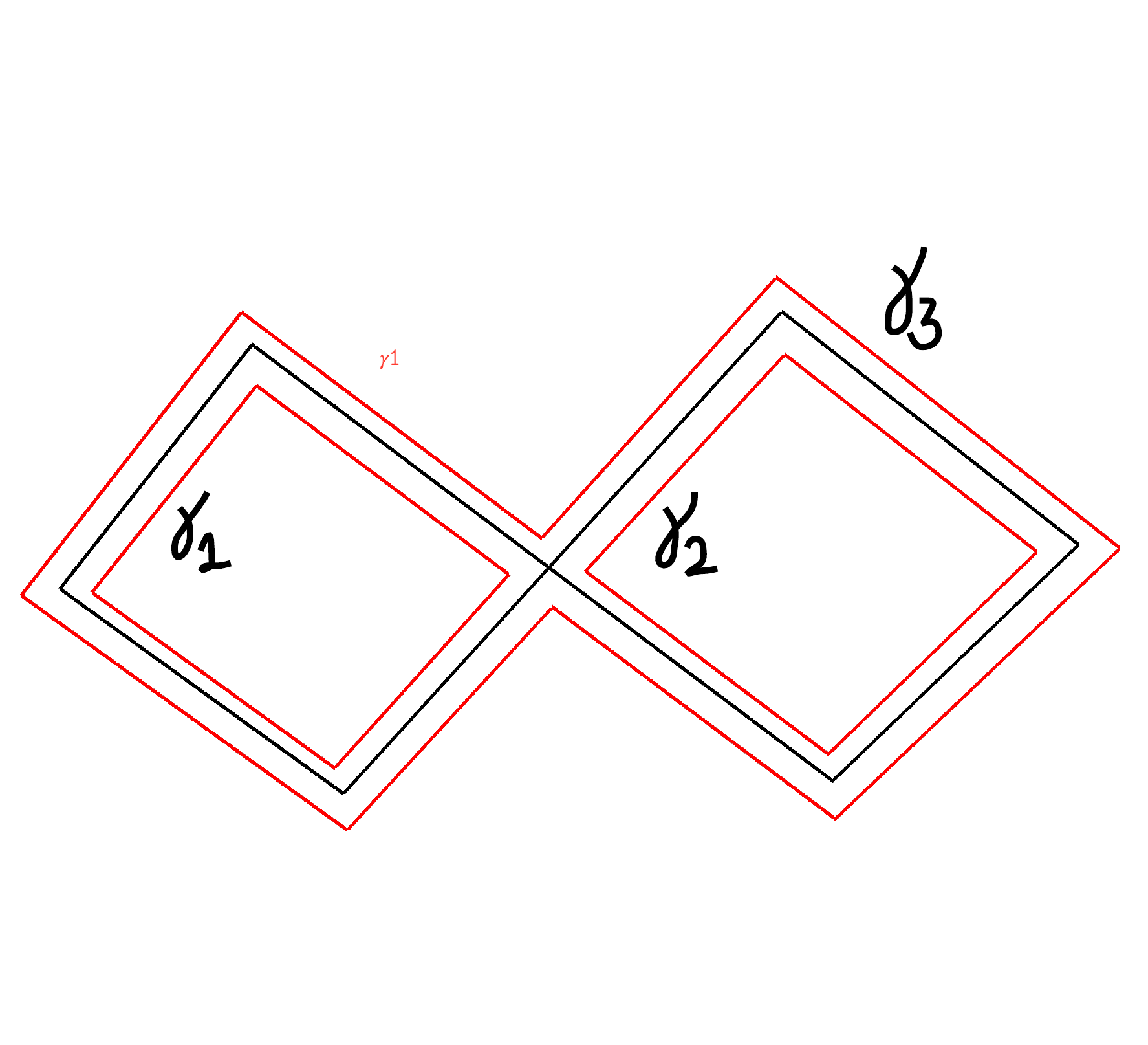}
\caption{Retractions of the three red loops to the spine, for clarity drawn slightly displaced from the graph itself. $\sfP$ is the set of two smaller red loops.}\label{GraphForSummaryWithPlaquettes}
\end{figure}

\item Let $d\mu_0$ denote the push-forward by $\bt_\bM$ of $\cD g$. By definition,
\begin{equation*}
\int \exp{\{\sum_{\gamma \in \sfP} \beta_\gamma \phi_\gamma (g)\}} \ \cD g=
\int \exp{\{\sum_{\gamma \in \sfP} \beta_\gamma \hphi (t_\gamma)\}} \ d\mu_0
\end{equation*}
where $\hphi$ is defined such that $\phi(\mathtt{g})=\hphi(t)$, with $trace(\mathtt{g})=2\cos{\pi t}$.

Of course, for this to be useful, one needs information on the image of the map $\bt_\bM$, as well as $d\mu_0$.

If $\#$ holds, it follows that in the limit of large $\Lambda$, the direct image measure $d\mu_0$ (corresponding to the expanded set $\tGamma$) tends (up to a nonzero constant) to the Lebesgue measure $d\bt \equiv \prod_{\gamma \in \tGamma} dt_\gamma$ restricted to $\cP_{\cT_\tGamma}$. In particular the integral in the expression (\ref{latticeintegral}), when $F=1$, can be meaningfully compared to:
\begin{equation}\label{trinionintegral}
\int_{\cP_{\cT_\tGamma}} \exp{\{\sum_{\gamma \in \tsfP} \beta_\gamma \hphi (t_\gamma)\}} \ d\bt
\end{equation}
It is clear that this can be extended to the case when $F$ is a function of the loops in $\tGamma$.

\item Jeffrey-Weitsman compute the volume of the polytope $\cP_{\cT_\tGamma}$ by applying the Verlinde formula to count integral points in (integral dilates) of $\cP_{\cT_\tGamma}$. In fact, we find that integration over the polytope can be efficiently enforced by a continuous analogue of the Verlinde algebra, At its heart is the following identity for the indicator function of the \emph{trinion polytope}. This polytope is the set $\cP \subset \bbR^3$ consisting of points $(t_1,t_2,t_3)$ subject to the spherical triangle inequalities:
\[
\begin{split}
0 \le t_i \le1,\ i=1,2,3,\ \ 
t_1+t_2+t_3 \le 2,\ and\ 
|t_2-t_3| \le t_1 \le t_2+t_3.
\end{split}
\]
We have the equality in $L^2([0,1]^3)$, proved in the Appendix \ref{trinionkernel}:
\begin{equation}\label{trinionpolytope}
\mathbf{1}_{\cP}(t_3,t_2,t_1)= \sum_{n=0}^\infty 
\frac{\sqrt{2}}{(n+1) \pi} \vu_n(t_1) \vu_n(t_2)  \vu_n(t_3)
\end{equation}
where
\[
\vu_n(t) \equiv \sqrt{2} \sin{(n+1) \pi t}, \ \ n=0,\dots ,
\]
is an orthonormal basis for $L_2[[0,1]]$.

(A version of this appears in Campbell Wheeler's 2016 Masters thesis: \texttt{https://guests.mpim-bonn.mpg.de/cjwh/files/MastersThesis.pdf})

In fact, the formalism can be used to compute volumes of (parabolic) moduli spaces. All this is in Appendix \ref {appendixvolumesetc}.

\item \emph{The trinion graph in our running example.} The three loops $\gamma_1,\gamma_2,\gamma_3$ in $\Lambda$ are the retractions of three loops $\tgamma_1,\tgamma_2,\tgamma_3$ in the genus 2 surface $\bS$. The latter cut the surface into \emph{two} pants $1,2$ sharing the three circles as common boundaries. The trinion graph $\cT_\tGamma$ is given below:

 \begin{figure}[H]
\includegraphics[scale=1,trim=250 300 250 300]{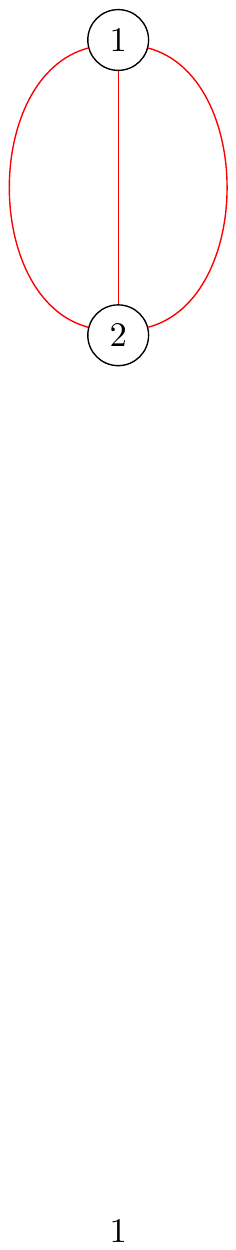}
\caption{The trinion graph of our running example.}\label{triniongraphofrunningexample}
\end{figure}

Let us compute the volume of the polytope:
\[
\begin{split}
&\sum_{n_1=0}^\infty \sum_{n_2=0}^\infty \frac{\sqrt{2}}{\pi (n_1+1)} \frac{\sqrt{2}}{\pi (n_2+1)}\\ \times 
&\int_0^1 \vu_{n_1}(t_1) \vu_{n_1}(t_2) \vu_{n_1}(t_3) \vu_{n_2}(t_1) \vu_{n_2}(t_2) \vu_{n_2}(t_3)\ dt_1 \ dt_2 \ dt_3\\
&=\sum_{n_1=0}^\infty \sum_{n_2=0}^\infty \frac{\sqrt{2}}{\pi (n_1+1)} \frac{\sqrt{2}}{\pi (n_2+1)} \delta_{n_1,n_2}\\
&=\sum_{n=0}^\infty  \frac{2}{\pi^2 (n+1)^2}\\
&=\frac{1}{3}
\end{split}
\]
This identification of indices corresponding to different vertices  will ocurr repeatedly.

\item In \S \ref{comparison} we compare the distribution of a single Wilson loop, as well as the joint distribution of two loops, w.r.to the symplectic and Haar (=lattice-theoretic) measures and show that they coincide in the large lattice limit. We also address constraints between loops -- ``nonabelian Gauss's law''. (Part of this discussion spills over to \S\ref{3dlattice}.)

\item In \S \ref{2dlattice}, we compute the symplectic versions of the partition function and Wilson loop for a two-dimensional periodic lattice, and discover that we recover the expressions of \cite{GrossWitten}. 

\item In \S \ref{3dlattice} we consider a three-dimensional periodic lattice, describe a natural pants decomposition of the surface $\bS$ that extends the plaquettes, draw the corresponding trinion graph, and proceed to write down 
expressions (as series) for the following:
\begin{itemize}
\item the symplectic partition function with the standard Wilson action, 
\item the symplectic partition function with the Migdal action,
\item the symplectic plaquette-plaquette correlation with the Migdal action, and
\item the symplectic 't Hooft loop expectation with the Migdal action,
\end{itemize}
The Migdal action \cite{Migdal,Witten} is a modification of the Wilson action which retains its qualitative features. In two dimensions it allows exact computation of the partition function w.r. to any triangulation of a surface; under any subdivision of the triangulation the partition function is preserved up to a coupling constant renormalisation. 

\item Here is the expression for the ``symplectic partition function'' with Migdal action:
\[
\begin{split}
Z'_{\be'} =
&\sum_{\substack{n_v=0 \\ v \in  \bbV_\Gamma}}^\infty \ \ \prod_{v \in \bbV_\Gamma}  (\frac{\sqrt{2}}{\pi (n_v+1)})^4\\
&\times \prod_{e \in \bbE_\Gamma}  \{\tsum_{l_e=|n_{v_1(e)}-n_{v_2(e)}|}^{n_{v_1(e)}+n_{v_2(e)}} \exp{\{-\frac{l_e(l_e+1){\be'}^2}{2}\}}(l_e+1)
\}
\end{split}
\]
Here $\Gamma$ is the dual lattice, $\bbV_\Gamma$ the set of its vertices, and $\bbE_\Gamma$ the set of edges. The notation $\tsum$ signals the omission of every alternate term in the sum. We have a slightly more elaborate expression for the symplectic partition function with the Wilson action. Each is a sum of positive terms. 
\end{enumerate}

The idea of relating lattice gauge theory on a graph to gauge theory on the associated surface has appeared elsewhere (albeit not in the context of the path integral) most notably in the works of A. Tyurin, for example \cite{Tyurin}. 

We close this Introduction with an important \textbf{caveat}. In the case of lattice gauge theories relevant to physics, in space-time dimensions two and three, we obtain expressions for quantities analogous to lattice-theoretic partition functions and expectation values of Wilson loops. We refer to these as ``symplectic partition functions'' and ``symplectic Wilson loop expectations''. Notwithstanding the impressive evidence, as yet we have no carefully formulated conjecture -- let alone theorems -- regarding the relationship between integrals over $\bM$ and over $\cM$.

\section{Background: Lattice gauge theory}\label{lattice}

Lattice gauge theory (\cite{Wilson}) provides finite-dimensional approximations to the Euclidean path integrals of gauge theory. The latter are ill-defined, and it is believed that these can be given a meaning as suitable limits of lattice gauge theory integrals, which we will define below. 

Lattice gauge theory has been intensely studied for over four decades by physicists by numerical methods. As for \emph{analytical} approaches, see {\cite{Chatterjee}}, which also serves as a survey of older work, both rigorous and heuristic.

\subsection{Generalities} Let $\Lambda$ be a finite connected graph. Let $\bbE=\bbE_\Lambda$ denote the set of edges. Let $\bbL=\bbL_\Lambda$ denote the set of edges together with a choice of orientation for each edge. Given an oriented edge $l$, we let $l'$ denote the same edge with orientation reversed.  Let $\bbV=\bbV_\Lambda$ denote the set of vertices.

Let $G$ be a compact Lie group, and let $\bA \subset G^\bbL$ denote the space of maps (``lattice gauge field configurations'') 
$$
\bbL \ni l \mapsto g_l \in G
$$
such that $g_{l'}=g_l^{-1}$.  Let $\bG=G^\bbV$, the space of maps (``lattice gauge transformations'') from $\bbV$ to $G$: a point of $\bG$ an assignation
$$
\bbV \ni v \mapsto h_v \in G
$$
The group $\bG$ acts on $\bA$:
$$
g \mapsto g^h; \ g^h_l=h_{head(l)} g_l h_{tail(l)}^{-1}
$$
where $tail(l),\ head(l)$ are the vertices at the tail and head of the oriented edge $l$. 

Upto a choice of orientation on each edge, $\bA=G^\bbE$, and the Haar measure (normalised so that the total volume is 1) on $G$ induces a gauge-invariant measure on $\bA$, independent of the choice. This measure we denote $\cD g$.

Suppose given an (oriented) path $\gamma$ in $\Lambda$, obtained by concatenating oriented edges $l_1,l_2,\dots,l_n$. Given a configuration $g \in G$,  its  \textit{(lattice) holonomy} $H_{\gamma}(g)$ along $\gamma$ is the product
$$
H_{\gamma}(g)=g_{l_n} g_{l_{n-1}} \dots g_{l_1}
$$
Note that $H_{\gamma}(g^h)=h(head(l_n))g_{l_n} g_{l_{n-1}} \dots g_{l_1}h^{-1}(tail(l_1))$

Fix a class-function $\phi$ on $G$. If $\gamma$ is closed, i.e., the head of $l_n$ is the  tail of $l_1$, and the map
$$
g \mapsto \phi_\gamma(g) \equiv \phi(H_\gamma(g))
$$
is independent of the base-point.  \emph{Note that the function $\phi_\gamma$ is gauge-invariant.} That is to say,
$$
\phi_\gamma(g^h)=\phi_\gamma(g)\ (g \in \bA, \ h \in \bG),\ 
$$
so that $\phi_\gamma$ descends to a function (for which we retain the same notation) on $\bA/\bG$. 

If $\phi$ is an \emph{even} class function, i.e., if $\phi(\mathtt{g})=\phi(\mathtt{g}^{-1}),\ \mathtt{g} \in G$, the functions $\phi_\gamma$ do not depend on the orientation of the closed path $\gamma$. We will consider only even class functions.

Lattice path integrals involve integrals of the form
$$
\int \exp{\{\sum_{\gamma \in \sfP} \beta_\gamma \phi_\gamma \}}F(g) \ \cD g
$$
Here $\cD g$ is the Haar measure on $\bA$ normalised so that $\int \cD g=1$,  $\sfP$ is a set of closed paths  in $\Lambda$, the $\beta_\gamma$ are real numbers, and $F$ is some gauge-invariant function.

From now on, unless otherwise flagged (when we will briefly consider the case $G=U(1)$), we will take $G$ to be $SU(2)$. Note that in this case the ring of real-algebraic class functions on $SU(2)$ is generated by the trace and \emph{all} class functions are even. The \emph{Wilson loop} is defined by:
\[
W_\gamma(g)=trace(H_\gamma(g))
\]
(In other words, $W_\gamma=trace_\gamma$.)
Note that $W_\gamma(g)$  is real-valued, $-2 \le W_\gamma(g) \le 2$, and $W_\gamma(g)$ is peaked around $g$ such that $H_\gamma(g)=Id$. There is a unique $t_\gamma(g) \in [0,1)$ such that 
\begin{equation}\label{deft}
W_\gamma(g)=2\cos(\pi t_\gamma(g))
\end{equation}

\subsection{Physically relevant lattice gauge theories} Let $d$ denote the dimension of space-time. (We will be concerned here with $d=2,3$.) In this context, by \textit{lattice} we mean the graph constructed as follows. First construct the infinite graph with vertices $\bbZ^d \subset \bbR^d$, and edges (`links') going from each vertex to its $2d$ nearest neighbours. Let $L$ be a large positive integer, and identify vertices $v_1,v_2$ if for every coordinate the difference is a multiple of $L$. Identify edges suitably. The resulting graph is  $\Lambda_{d,L}$. We will usually suppress the subscript and set $\Lambda=\Lambda_{d,L}$. 

A \textit{plaquette} $P$ is the boundary of a unit square whose corners are vertices of $\Lambda$ and whose edges are parallel to the coordinate axes in $\bbR^d$.  We will think of plaquettes as closed unoriented paths.

The measure of interest is
$$
\exp{\{{\frac{2}{\be^2} \sum_{plaquettes\ P}  W_P(g) }\}} \ \cD g
$$
where $e$ is a positive constant. In terms of the discussion in the Introduction, we have taken for $\sfP$ the set of plaquettes. The function
\[
g \mapsto -\frac{2}{\be^2} \sum_{plaquettes\ P}  W_P(g)
\]
is the ``Wilson action''.

\begin{figure}[H]
\includegraphics[scale=0.3,trim=20 -40 20 20]{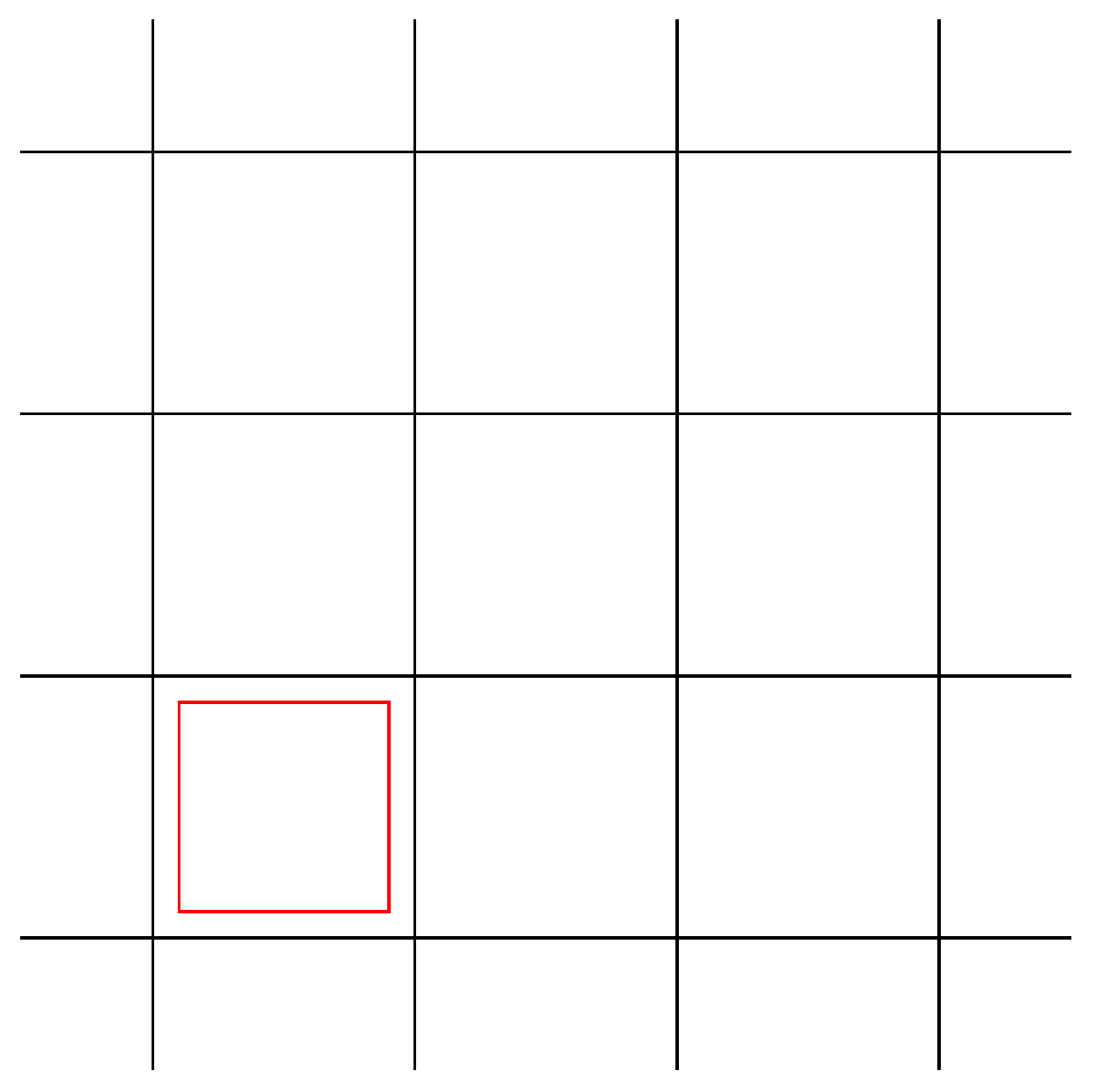}
\caption{(Part of) a $d=2$ lattice, with one plaquette $P$ (in red) shown, drawn displaced from the graph for clarity.}\label{2DLattice}
\end{figure}

 \subsection{The associated handlebody and surface}\label{pumping}

``Pumping up''  the lattice, we get a handlebody and surface as in the Introduction.

\begin{figure}[H]
  \includegraphics[scale=0.5,trim=20 20 20 20]{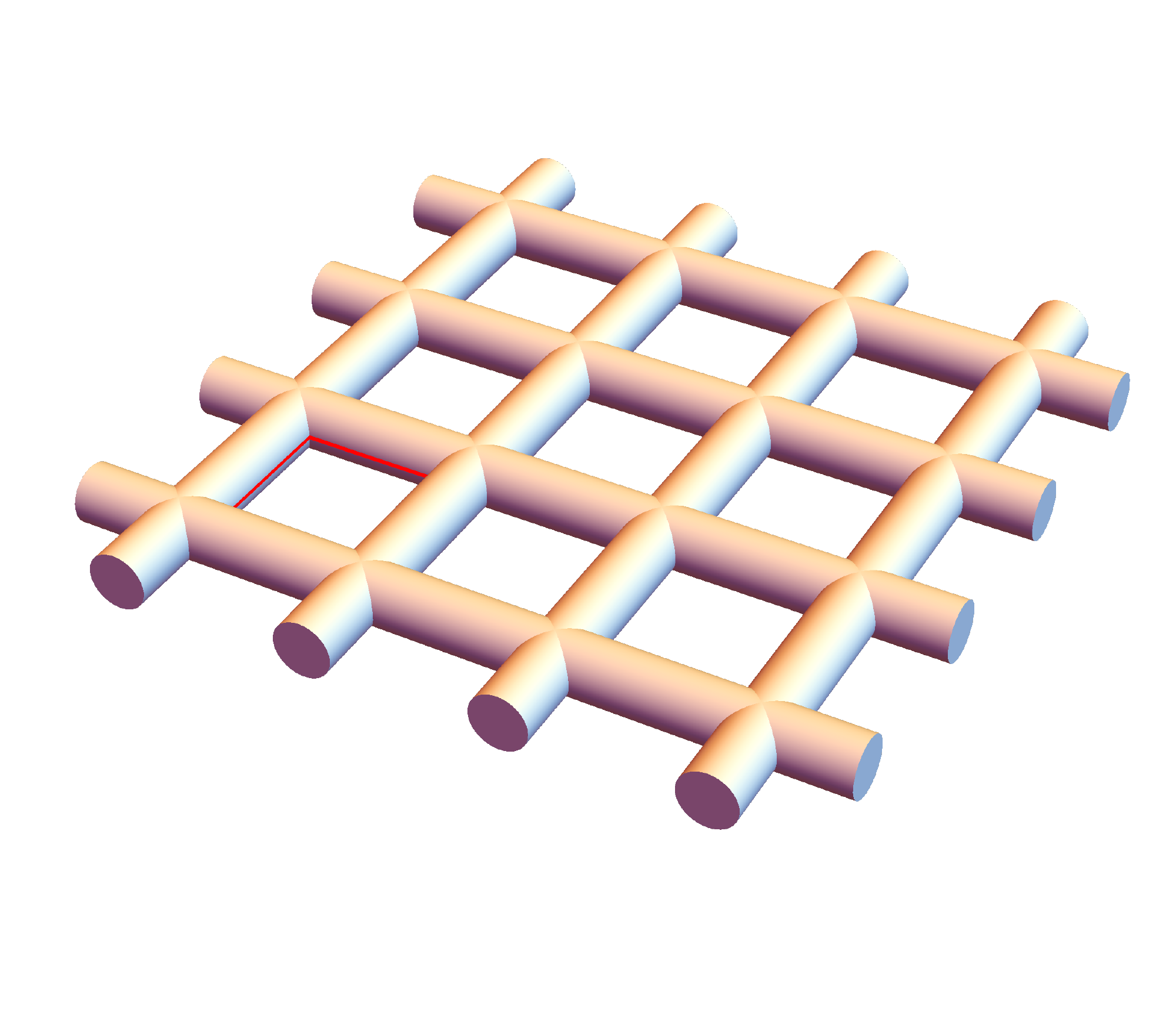}
    \caption{(Part of) a $d=2$ ``lattice of tubes'', with a loop $\tP$ (in red) that retracts to the plaquette $P$ of Figure \ref{2DLattice}.}\label{2DTubes}
    \end{figure}

Next to some numerology. Given a hypercubical lattice $\Lambda$ in $d$ space-time dimensions, with periodic boundary conditions of period $L$ in each direction, the total number of vertices  of $\Lambda$ is $L^d \equiv N$.  There are  $^dC_2 N$ plaquettes.  The genus $\sfb$ of the corresponding surface $\bS$ and $N$ are related by
$$
(d-1)N=\sfb-1
$$

\begin{remark} A compact oriented surface is determined up to diffeomorphism by its genus. Our constructions, however, break the diffeomorphism invariance. Specifically, we retain the invariance under the subgroup $Diff^0 \ \bS$ of diffeomorphisms connected to the identity, but break the invariance under the modular group $Diff \ \bS/Diff^0 \ \bS$ to a subgroup of translations isomorphic to $(\bbZ/L\bbZ)^d$.
\end{remark}

\section{Background: The Character variety, Wilson loops, and the Goldman bracket}\label{CharVar}

\subsection{Surfaces, pants decompositions.} 

We recall some basic constructions whose natural context is the hyperbolic geometry of surfaces and three-manifold topology. Here we are concerned mostly with elementary topological aspects. 
 
A ``pair of pants'' or ``trinion'' is a two-manifold with smooth boundary diffeomorphic to a 2-sphere with three (disjoint) open discs removed from it. Given a genus $\sfb$ surface $\bS$, we can decompose it into (mutually disjoint except for boundary circles) $2(\sfb-1)$ trinions, identifying the boundary circles two at a time; this is called a ``pair-of-pants decomposition''.  The boundaries of the pants consist of $3(\sfb-1)$  mutually disjoint embedded circles $\tGamma=\{\tgamma_1,\dots,\tgamma_{3(\sfb-1)}\}$ which determine the decomposition; we will talk of a ``pants decomposition $\tGamma$''.  It is a fact that given a compact oriented surface and mutually disjoint embedded circles such that none of them bounds a disc and no two are freely isotopic, they can be extended to be part of (the boundaries of) a pants decomposition.

Associated to such a decomposition $\tGamma$ are

\begin{enumerate}
\item a connected trivalent graph (``the trinion graph'') ${\cT_\tGamma}$ in which each vertex represents one of the trinions in the decomposition, and two distinct vertices are joined by an edge if the two corresponding trinions share a boundary circle, and an edge connects a vertex to itself if two boundary components  of the corresponding trinion get identified on $\bS$. \emph{We shall only consider pants decompositions such that $\cT_\tGamma$ has no loops and cannot be disconnected by removing one edge.
}\item a handlebody (i.e., a 3-manifold which is a connected sum of solid tori) $M_{\cT_\tGamma}$ with boundary $\bS$, such that each of the boundary circles $\tgamma_i$ bounds a disc embedded in $M_{\cT_\tGamma}$. The handlebody $M_{\cT_\tGamma}$ can be realised as a closed regular neighbourhood of an  embedding of ${\cT_\tGamma}$ in any three-manifold, say, $\bbR^3$.
\end{enumerate}

\underline{Remark/Warning}: Note that in the context of our running example of the graph $\Lambda$, the handlebody $M_{\cT_\tGamma}$ corresponding to the pants decomposition $\tGamma$ (with graph the of Figure \ref{triniongraphofrunningexample}) is \emph{not} $M_\Lambda$. For one, the loops $\tgamma_1,\tgamma_2,\tgamma_3$ do not bound discs in $M_\Lambda$. 

\subsection{The moduli space of flat connections} Let $\bS$ be a compact oriented surface of genus $\sfb$ without boundary. Let  $\cM$  denote the set of gauge-equivalence classes of flat connections on the trivial  $SU(2)$-bundle on $\bS$, or equivalently the space of flat $\mathfrak{su}(2)$-valued gauge-fields  modulo gauge transformations. We will let $\cM^s$ denote the gauge-equivalence classes of \emph{irreducible} $SU(2)$ connections. (A $G$-connection is irreducible if the only gauge transformations that leave it unchanged are the constant ones with values in the centre of $G$.) The set $\cM^s$ is naturally a manifold of dimension  $6(\sfb-1)$, and by results of  \cite{AtiyahBott} and \cite{Goldman} it carries a natural symplectic structure.

There are well-known bijections
\begin{equation*}
\begin{split}
\cM&=Rep(\pi_1(\bS),SU(2))/SU(2)\\
\cM^s&=IrRep(\pi_1(\bS),SU(2))/SU(2)\\
\end{split}
\end{equation*}
where $\pi_1(\bS)$ is the fundamental group of $\bS$, based at a point $x_0$ (which we omit from the notation), $Rep$ (resp., $IrRep$) denotes the 
set of representations (resp., irreducible representations) into $SU(2)$ and the notation $/SU(2)$ signals taking the quotient by the adjoint action. The spaces on the right -- these are the character varieties of the title --  have natural structures of real affine algebraic varieties. For any class function $\phi$, a closed path $\tgamma$ in $\bS$, and a flat $SU(2)$ connection $A$ we define $\phi_\tgamma(A)$ to be $\phi(H_\gamma(A))$ where $H_\gamma(A)$ is the holonomy of $A$ around $\tgamma$ at any base point. In particular the Wilson loop functions $W_\tgamma$ correspond 
to the trace.

The Wilson loop functions $W_\tgamma$ are real-valued and algebraic, and in fact \cite{CharlesMarche} generate the algebra of such functions.

\subsection{Poisson brackets of Wilson loops}  Poisson brackets between the functions $W_\tgamma$ were computed by W. Goldman. We recall the formulae in the cases of interest. Let $\tgamma$ and $\teta$ be two closed \emph{oriented} loops.
\begin{itemize}
\item If $\tgamma$ and $\teta$ are disjoint, or can be isotoped to be so, then $\{W_\tgamma,W_\teta\}=0$, and
\item If $\tgamma$ and $\teta$ intersect transversally, then
\begin{equation}\label{GoldmanBrackets}
\{W_\tgamma,W_\teta\}=\sum_{p \in \tgamma \cap \teta}
\epsilon_p(\tgamma,\teta) \big(W_{\tgamma \#_p \teta}-\frac{1}{2}W_\tgamma W_\teta\big)
\end{equation}
Here $\epsilon_p(\tgamma,\teta)=\pm 1$ is the oriented intersection number of $\tgamma$ and $\teta$ at $p$, and $\tgamma \#_p \teta$ is the closed loop obtained by concatenating $\tgamma$ and $\teta$ at $p$.
\end{itemize}

Note that the intersection number changes sign if the orientation on \emph{one} of the loops is reversed: $\tgamma \mapsto \tgamma'$. This is consistent because on the one hand the Wilson loop function does not change if the loop flips orientation, but
\[
W_{\tgamma \#_p \teta}+W_{\tgamma' \#_p \teta}=W_\tgamma W_\teta
\]
because of the identity
\[
trace(\mathtt{g}_1\mathtt{g}_2 )+trace(\mathtt{g}^{-1}_1\mathtt{g}_2 )=trace(\mathtt{g}_1)trace(\mathtt{g}_2)
\]
that holds for  $\mathtt{g}_1, \ \mathtt{g}_2$ in $SU(2)$.

As a consequence, the Hamiltonian vector fields generated by Wilson loop functions corresponding to disjoint loops commute. Goldman gave a geometric description of these flows, and this has the consequence, beautifully exploited by L. Jeffrey and J. Weitsman, that a pants decomposition of $\Sigma$ yields a Hamiltonian action of the compact torus $U(1)^{3(\sfb-1)}$ on $\cM$, the map $\bt$ being essentially the moment map. 

Although the hypotheses of Duistermaat-Heckman theorem do quite hold (because of the singularities of $\cM$), its conclusions (as summarised in item \ref{DH} of the Introduction) hold. For a clear account, see \cite{Audin}.

\section{Comparison of integrals on $\bM$ and $\cM$}\label{comparison}

Given a closed, non-self-interesting path $\gamma$ on the lattice $\Lambda$, and any real-valued continuous function $f_1$ on $[0,\pi]$, one can check:
\[
\begin{split}
\int f_1(\pi t_\gamma(g)) \cD g&=\int f_1(\arccos{\frac{1}{2}trace{H_\gamma(g)}}) \cD g\\
&\underset{*}=\int_G f_1(\arccos{ \frac{1}{2}trace(\mathtt{g})}) d\mathtt{g}\\
&=\int_0^1 f_1(\pi t) \tdt
\end{split}
\]
 where we have introduced the notation $\tdt=2\sin^2{\pi t}\ dt$. The equality (*) is obtained by using gauge transformations to fix all the link variables except one along $\gamma$ to $Id$. This reasoning will be used repeatedly below without mention. 
 
 One can also check that given \emph{two} disjoint closed non-self-interesting loops $\gamma_1,\gamma_2$ the corresponding Wilson loops are independent (as random variables).
 
\subsection{Distribution of one Wilson loop.} Let us consider a simple closed curve that is neither contractible nor separating. Without loss of generality we can suppose that it is one of the bounding circles (say, $\tgamma_1$) of a pants decomposition $\tLambda$ with associated trinion graph $\cT_\tLambda$. Given $f_1$ on $[0,\pi]$, consider the ratio of integrals
\[
I_{f_1} \equiv \frac{\int_{\cP_{\cT_\tLambda}} f_1(\pi t_1)d\bt}{\int_{\cP_{\cT_\tLambda}} 1 d\bt}
\]
We use the identity \ref{trinionpolytope} and conclude that
\[
\begin{split}
I_{f_1}&= \frac{\sum_{n=0}^\infty 
(\frac{\sqrt{2}}{(n+1) \pi})^{2(\sfb-1)} \int_0^1 f_1(\pi t_1) \vu_n^2 (t_1) dt_1
}{\sum_{n=0}^\infty 
(\frac{\sqrt{2}}{(n+1) \pi})^{2(\sfb-1)}}\\
&= \frac{\sum_{n=0}^\infty 
(\frac{\sqrt{2}}{(n+1) \pi})^{2(\sfb-1)} \int_0^1 f_1(\pi t_1) 2\sin^2{(n+1) \pi t_1} dt_1
}{\sum_{n=0}^\infty 
(\frac{\sqrt{2}}{(n+1) \pi})^{2(\sfb-1)}}
\end{split}
\]
If $\sfb=2$ (as in our running example), we get
\[
\begin{split}
I_{f_1}&= \frac{\sum_{n=0}^\infty 
(\frac{\sqrt{2}}{(n+1) \pi})^{2} \int_0^1 f_1(\pi t_1) 2\sin^2{(n+1) \pi t_1} dt_1
}{\sum_{n=0}^\infty 
(\frac{\sqrt{2}}{(n+1) \pi})^{2}}\\
&= \frac{\sum_{n=0}^\infty 
(\frac{\sqrt{2}}{(n+1) \pi})^{2} \int_0^1 f_1(\pi t_1) (2\sin^2{(n+1) \pi t_1}-1)dt_1
}{\sum_{n=0}^\infty 
(\frac{\sqrt{2}}{(n+1) \pi})^{2}}\\
&+ \frac{\sum_{n=0}^\infty 
(\frac{\sqrt{2}}{(n+1) \pi})^{2} \int_0^1 f_1(\pi t_1) dt_1
}{\sum_{n=0}^\infty 
(\frac{\sqrt{2}}{(n+1) \pi})^{2}}\\
&= -\frac{\sum_{k=1}^\infty 
\frac{2}{k^2 \pi^2} \int_0^1 f_1(\pi t_1) \cos{2k \pi t_1}dt_1
}{\sum_{n=0}^\infty 
\frac{2}{k^2 \pi^2}}\\
&+\int_0^1 f_1(\pi t_1) dt_1
\\
&= \int_0^1 f_1(\pi t_1)\{1-\frac{B_2(t_1)}{b_2}\}dt_1\\
&= \int_0^1 f_1(\pi t_1)6(t_1-t_1^2)dt_1\\
\end{split}
\]
where
$$
B_{2}(x)=\frac{1}{\pi^{2}} \sum_{k=1}^\infty \frac{\cos(2 \pi k x)}{k^{2}} = x^2-x+1/6 
$$
is a Bernoulli polynomial

This can be verified directly: 
\[
\begin{split}
I_{f_1}&=\frac{\int_{\cP} f_1(\pi t_1)dt_1dt_2dt_3}{\int_{\cP} dt_1dt_2dt_3}\\
&=\frac{\int_0^1 f_1(\pi t_1) area(\{(t_2,t_3)|(t_1,t_2,t_3) \in \cP\}) dt_1}{\int_0^1 area(\{(t_2,t_3)|(t_1,t_2,t_3) \in \cP\}) dt_1}\\
&=\frac{\int_0^1 f_1(\pi t_1)2t_1(1-t_1) dt_1}{\int_0^1 2t_1(1-t_1) dt_1}\\
&= \int_0^1 f_1(\pi t_1)6(t_1-t_1^2)dt_1
\end{split}
\]

On the other hand, as $\sfb \to \infty$, we get
\[
I_{f_1} \to \int_0^1 f_1(\pi t_1) 2\sin^2 {\pi t_1} dt_1=\int_0^1 f_1(\pi t_1)\ \tdt
 \]
For a more careful treatment of the limit, see Appendix \ref{appendixvolumesetc}.

\subsection{Joint distribution of two disjoint Wilson loops.} Given two disjoint loops $\tC_1,\tC_2$ and given a function $f_2$ consider the ratio
\[
I_{f_2} \equiv \frac{\int_{\cP_{\cT_\tLambda}} f_2(\pi t_1,\pi t_2)d\bt}{\int_{\cP_{\cT_\tLambda}} 1 d\bt}
\]
This time we need to assume that the loops $\tC_1,\tC_2$ are a nonseparating pair, in other words, that $\bS \setminus (\tC_1\cup \tC_2)$ is connected.
As before we conclude
\[
I_{f_2}= \frac{\sum_{n=0}^\infty 
(\frac{\sqrt{2}}{(n+1) \pi})^{2(\sfb-1)} \int f_2(\pi t_1,\pi t_2)\ 2\sin^2{(n+1) \pi t_1} \ 2\sin^2{(n+1) \pi t_2} \ dt_1 dt_2}
{\sum_{n=0}^\infty 
(\frac{\sqrt{2}}{(n+1) \pi})^{2(\sfb-1)}}
\] 
As $\sfb \to \infty$, we get
\[
I_{f_2} \to \int_0^1 \int_0^1 f_2(\pi t_1,\pi t_2)\ \tdt_1\ \tdt_2
\]
which shows that in the limit of large graphs, disjoint Wilson loops are independent.

\subsection{``Nonabelian Gauss's Law'' I} First some preliminaries.  Recall that the characters (traces of irreducible representations) of $SU(2)$ are the functions
\[
SU(2) \ni \mathtt{g} \mapsto \chi_l(\mathtt{g})=\frac{\sin{(l+1) \pi t_\mathtt{g}}}{\sin{\pi t_\mathtt{g}}},\ l=0,1,
\]
where $trace(\mathtt{g})=\chi_1(\mathtt{g})=2\cos{\pi t_\mathtt{g}}$. We will define functions $\hchi_l:[0,\pi]\to \bbR$ by
\[
\hchi_l(u)=\frac{\sin{(l+1)u}}{\sin{u}}
\]
Note that $\vu_l(t)=\sqrt{2} \hchi_l(\pi t) \sin{\pi t}$.

Consider three loops $\tC_1,\tC_2,\tC_3$ in $\bS$ such that $\bS \setminus (\tC_1\cup \tC_2 \cup \tC_3)$ is the disjoint union of a sphere with three holes and a connected surface $\bS'$, also with three holes. Retracting to the spine $\Lambda$, these yield three loops $C_1,C_2,C_3$ such that (for suitable choices of orientations) for any lattice gauge field configuration $g$, the lattice holonomies $\mathtt{h}_1(g),\mathtt{h}_2(g),\mathtt{h}_3(g)$ obey $\mathtt{h}_1(g) \mathtt{h}_2(g)=\mathtt{h}_3(g)$. Given three characters $\chi_{l_i}, i=1,2,3$, the lattice gauge theory integral of the product of the corresponding functions $(\chi_{l_i})_{\gamma_i}$ obeys
\[
\begin{split}
\int_{\bA} (\chi_{l_1})_{\gamma_1} (g) (\chi_{l_2})_{\gamma_2} (g) (\chi_{l_3})_{\gamma_3} (g)  \cD g&=\int_{\bA} \chi_{l_1} (\mathtt{h}_1(g)) \chi_{l_2} (\mathtt{h}_2(g)) \chi_{l_3} (\mathtt{h}_3(g)) \cD g\\
&=\int_{G \times G} \chi_{l_1} (\mathtt{h}_1) \chi_{l_2} (\mathtt{h}_2) \chi_{l_3} (\mathtt{h}_1\mathtt{h}_2) d\mathtt{h}_1 d\mathtt{h}_2\\
&=\frac{1}{l_1+1}\delta_{l_1,l_2,l_3}\\
\end{split}
\]
using the property: $\int_G \chi_l(\mathtt{g}_1\mathtt{h})\chi_m(\mathtt{h}^{-1}\mathtt{g}_2) d\mathtt{h}=\frac{\delta_{l,m}}{l+1}\chi_l(\mathtt{g}_1\mathtt{g}_2)$.

Choose pants decomposition such that the above sphere-with-three-holes is one of the pants in the decomposition. The trinion graph, in the neighbourhood of the vertex representing this pair of pants, looks like this:

\begin{figure}[H]
  \includegraphics[scale=.7,trim=90 250 50 120]{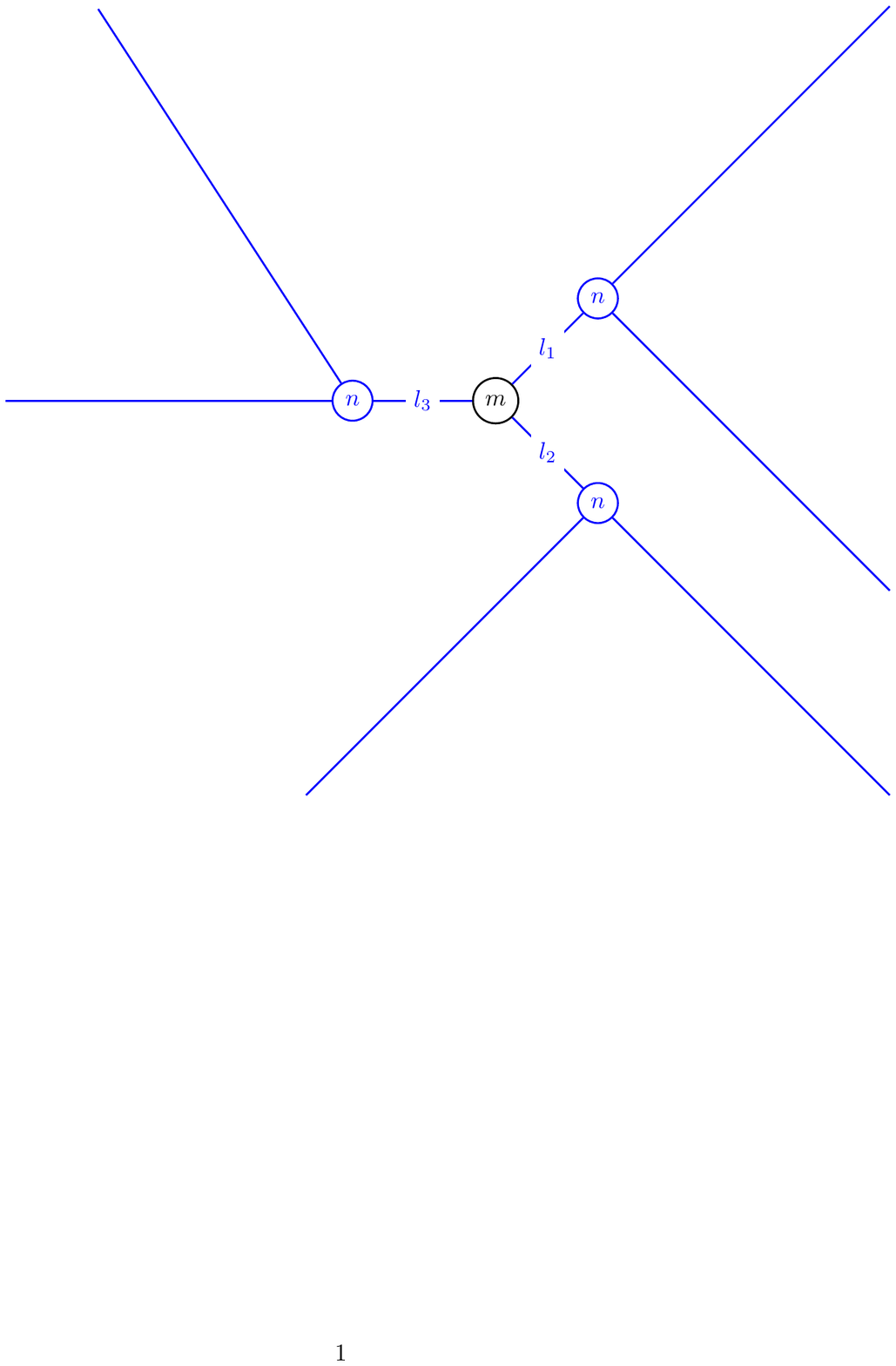}
    \caption{The trinion graph locally around the vertex corresponding to the sphere-with-three-holes; the labels on the edges and vertices correspond to indices in Equation (\ref{ThreePlaquettesEquation}).}
    \label{TrinionThreePlaquettes}
    \end{figure} 

The trinion integral, therefore, yields:
\begin{equation}\label{ThreePlaquettesEquation}
\begin{split}
&\frac{1}{\sum_{n=0}^\infty (\frac{\sqrt{2}}{(n+1) \pi})^{2(\sfb-1)}}\\
\times &\{
\sum_{n=0}^\infty (\frac{\sqrt{2}}{(n+1) \pi})^{2(\sfb-1)-1}  \sum_{m=0}^\infty (\frac{\sqrt{2}}{(m+1) \pi})\\
&\int \hchi_{l_1} (\pi t_1) \hchi_{l_2} (\pi t_2) \hchi_{l_3} (\pi t_3) \vu_n(t_1) \vu_m (t_1) \vu_n(t_2) \vu_m (t_2) \vu_n(t_3) \vu_m (t_3) dt_1 dt_2 dt_3\}
\end{split}
\end{equation}
As $\sfb \to \infty$, this tends to
\[
\begin{split}
&\frac{\pi}{\sqrt{2}}
\times \{
\sum_{m=0}^\infty (\frac{\sqrt{2}}{(m+1) \pi}\\
&\int \hchi_{l_1} (\pi t_1) \hchi_{l_2} (\pi t_2) \hchi_{l_3} (\pi t_3) \vu_1(t_1) \vu_m (t_1) \vu_1(t_2) \vu_m (t_2) \vu_1(t_3) \vu_m (t_3) dt_1 dt_2 dt_3\}\\
&=\frac{\pi}{\sqrt{2}}
\times \{
\sum_{m=0}^\infty (\frac{\sqrt{2}}{(m+1) \pi}\\
&\int \hchi_{l_1} (\pi t_1) \hchi_{l_2} (\pi t_2) \hchi_{l_3} (\pi t_3)  \hchi_m (\pi t_1) \hchi_m (\pi t_2)  \hchi_m (\pi t_3)\  \tdt_1 \ \tdt_2 \ \tdt_3\}\\
&=\frac{1}{l_1+1} \delta_{l_1,l_2,l_3}\\
\end{split}
\]

We will consider another case of the ``Nonabelian Gauss's Law'' below.

\section{Symplectic Lattice Theory, $d=2$}\label{2dlattice}

\subsection{Preliminaries} We begin with a pants decomposition of $\bS$ such that the lifted plaquettes $\tP_i,i=1,\dots,N$ are among the boundary circles. 

 \begin{figure}[H]
  \includegraphics[scale=0.5,trim=20 20 20 20,clip]{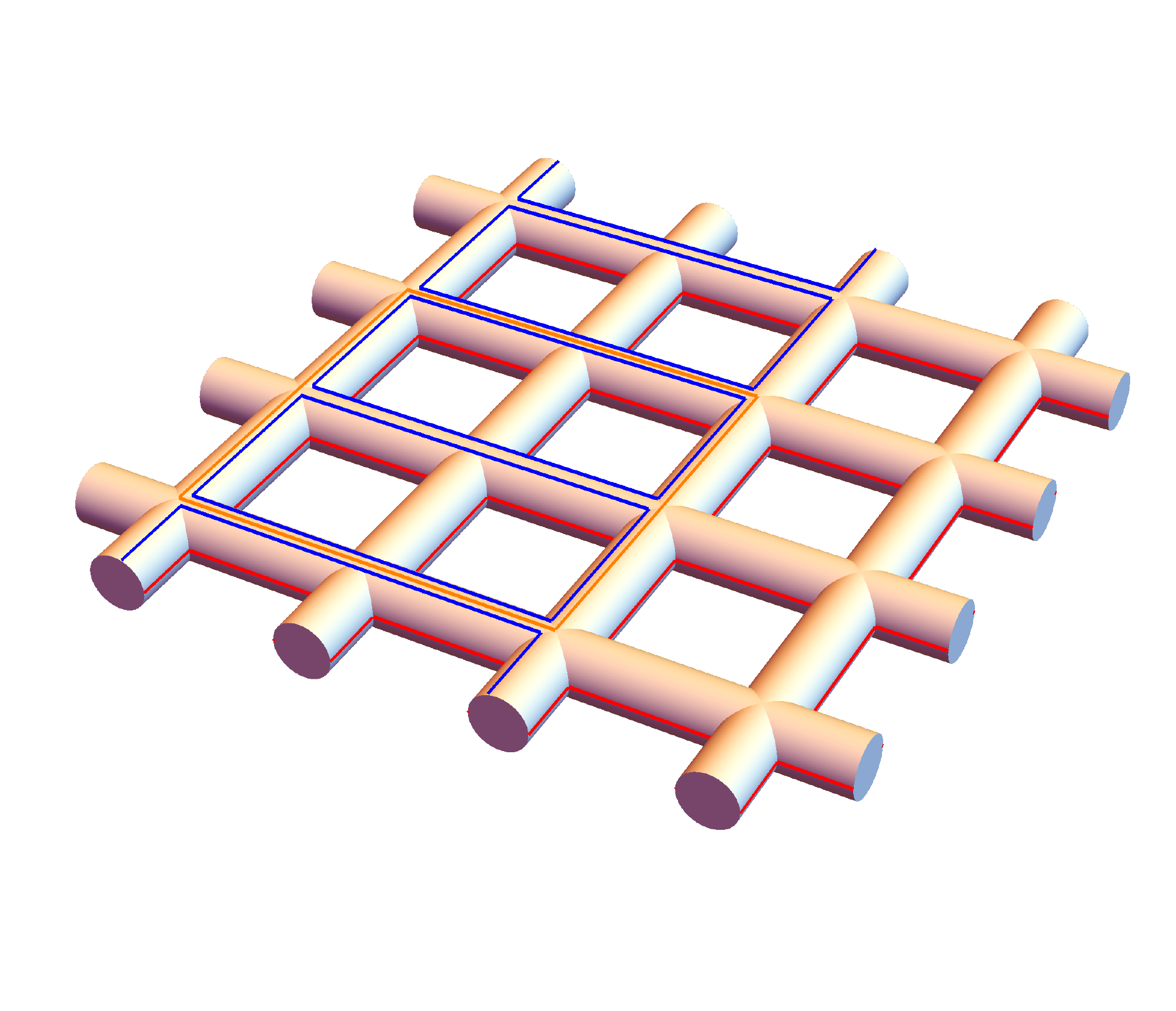}
    \caption{(Part of) a pants decomposition extending the plaquette loops (in red).}\label{2DTubesPlaquetteTrinions}
    \end{figure}

We suppose that $N$ is a power of 2:
$$
N=2^{2n}
$$
Cut $\bS$ along each lifted plaquette. This yields two tori, which we label $\bS'$ and $\bS''$, each with $N$ discs cut out, the boundaries being copies $\tP'_i$ and $\tP''_i$ of the $\tP_i$. Each of the two tori has a pants decomposition (illustrated below in the case $n=3$ -- the opposite sides of the rectangle are identified to give a torus) into 
$$
2\{2^{2n-1}+2^{2n-2}+\dots+2+2\}=2\{2(2^{2n-1}-1)+2\}=2^{2n+1}=2N=2(\sfb-1)
$$
trinions. 
 \begin{figure}[H]
  \includegraphics[scale=0.7,trim=140 120 100 120,clip]{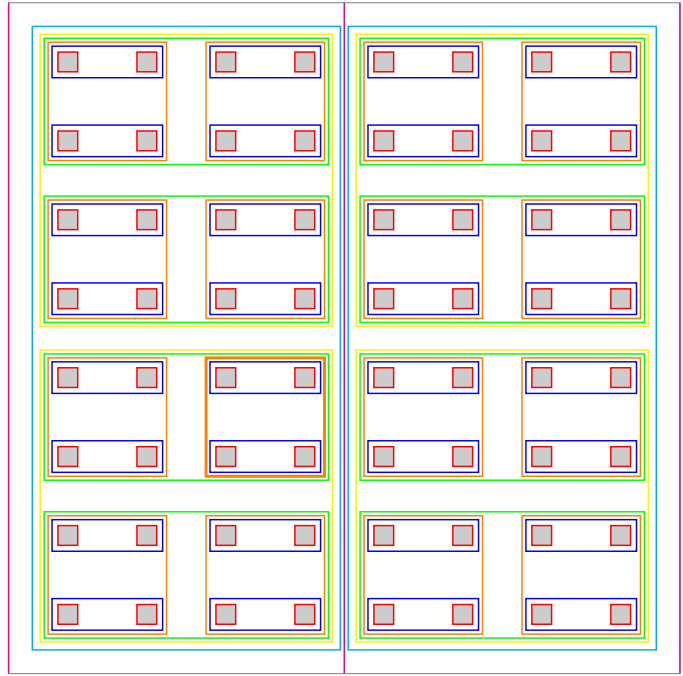}
    \caption{(Part of) a pants decomposition extending the plaquette loops (in red).}\label{2DTubesPlaquetteTrinions}
    \end{figure}   
This gives $2 \times 2^{2n}=2N=2(\sfb-1)$ trinions which is as it should be.
There are $2\times 2+2\{2^{n}+\dots+2\}-2^{n}=4+4(2^{n}-1)-2^{n}=3 \times 2^n=3N=3(\sfb-1)$ boundary circles, among them the $N$ plaquettes. There is one subtlety, though. If the two tori are glued in the obvious way, the $N$ ``extra'' plaquettes from each of them will be identified in pairs under the retraction to $\Lambda$. To avoid this we glue the tori after translation in (say) the $x$-direction by one unit.

 \begin{figure}[H]
  \includegraphics[scale=0.75,trim=120 160 60 120,clip]{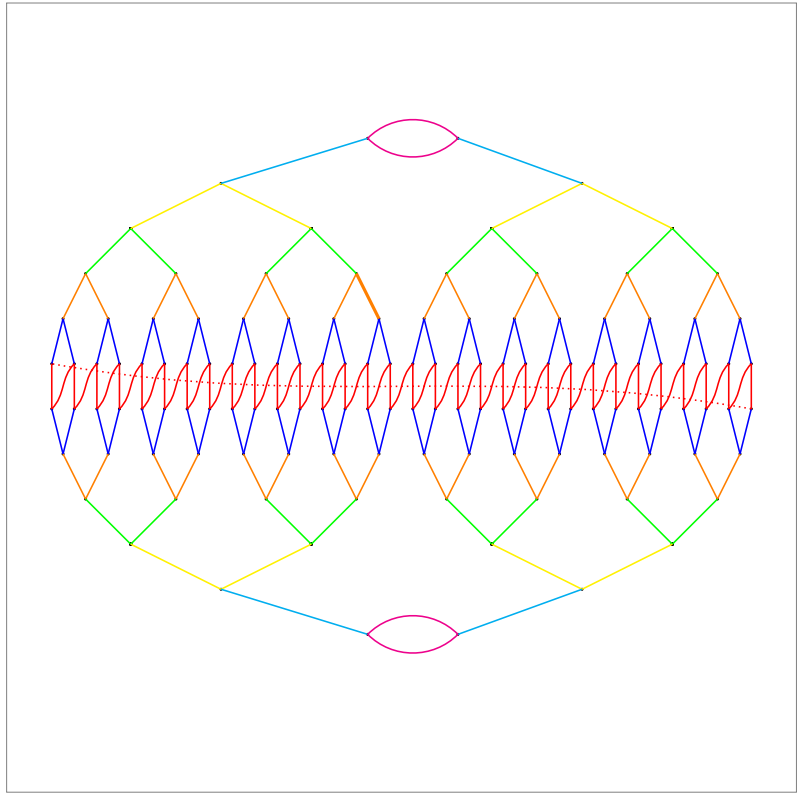}
    \caption{The trinion graph corresponding to the pants decomposition in Figure \ref{2DTubesPlaquetteTrinions}).}\label{2DTrinionGraph}
    \end{figure}

\subsection{Partition function and Wilson Loop}

The symplectic partition function evaluates to:
\[
Z_\be \equiv \sum_{l,m \ge 0} \{\frac{\sqrt{2}}{\pi (l+1)}\}^{N} \{\frac{\sqrt{2}}{\pi (m+1)}\}^{N} \{\int_0^1 \vu_l(t) \vu_m(t) \exp{(\frac{4}{\be^2} \cos{\pi t})}dt\}^N\\
\]
and
\[
\begin{split}
\frac{Z_\be}{Z_\infty}  \underset{large\ N}\sim\ \  &\{\int_0^1 \exp{(\frac{4}{\be^2} \cos{\pi t})}2 \sin^2{\pi t}\ dt\}^N\\
&=\{\int_0^1 \exp{(\frac{4}{\be^2} \cos{\pi t})}\ \tdt\}^N=z(\be^2,2)^N
\end{split}
\]
in agreement with (and, with minor changes,  using the notation of) the expression for the classic lattice partition function, \cite{GrossWitten}.

Let us compute symplectic expectation values of  Wilson loops. Consider a loop $C$ of length $4N'=4\times 2^{n'}$ in our pants decomposition . (For example, in our example, the thick orange line represents a loop of length 8.) We have
\[
\begin{split}
<W_C> &\equiv \frac{1}{Z_\be} \sum_{l,l',m \ge 0} \{\frac{\sqrt{2}}{\pi (l+1)}\}^{N-N'} \{\frac{\sqrt{2}}{\pi (l'+1)}\}^{N'}\{\frac{\sqrt{2}}{\pi (m+1)}\}^{N}\\ &\{\int_0^1 \vu_l(t) \vu_m(t) \exp{(\frac{4}{\be^2} \cos{\pi t})\ dt}\}^{N-N'}\{\int_0^1 \vu_{l'}(t) \vu_m(t) \exp{(\frac{4}{\be^2} \cos{\pi t})\ dt}\}^{N'}\\
 \\ &\times \int_0^1 \vu_l(t)\vu_{l'}(t) \ 2\cos{\pi t}dt
\end{split}
\]
Note that
\[
\begin{split}
\int_0^1 \vu_0(t)\vu_{l'}(t) \ 2\cos{\pi t}dt&=\int_0^1 2 \sin{\pi t}\ \sin{(l'+1)\pi t}\ 2\cos{\pi t} dt\\
&=2\int_0^1  \sin{2 \pi t} \sin{(l'+1)\pi t}\ dt=\delta_{l',1}
\end{split}
\]
so that
\[
\begin{split}
<W_C> &\underset{large\ N}\sim\ \ 
 \frac{\sum_{l' \ge 0} \frac{1}{(l'+1)^{N'}}\{\int_0^1 \exp{(\frac{4}{\be^2} \cos{\pi t})}2 \sin{(l'+1)\pi t}\sin{\pi t}\ dt\}^{N'}}{\{\int_0^1 \exp{(\frac{4}{\be^2} \cos{\pi t})}2 \sin^2{\pi t}\ dt\}^{N'}}\\
&\times \int_0^1 2\sin{\pi t} \sin{(l'+1) \pi t}\ 2\cos{\pi t}dt\\
&=\frac{1}{2^{N'}}\frac{\{\int_0^1 \exp{(\frac{4}{\be^2} \cos{\pi t})}2 \sin (2\pi t)\sin{\pi t}\ dt\}^{N'}}{\{\int_0^1 \exp{(\frac{4}{\be^2} \cos{\pi t})}2 \sin^2{\pi t}\ dt\}^{N'}}\\
&=\frac{1}{2^{N'}}\frac{\{\int_0^1 \exp{(\frac{4}{\be^2} \cos{\pi t})}2\cos (\pi t)\ \tdt\}^{N'}}{\{\int_0^1 \exp{(\frac{4}{\be^2} \cos{\pi t})}\ \tdt\}^{N'}}\\
&=\{w(\be^2,2)\}^{N'}
\end{split}
\]
again in agreement with (and using the notation of) \cite{GrossWitten}.

\section{Symplectic Lattice Theory, $d=3$}\label{3dlattice}

\subsection{Preliminaries} 

In this case, the lattice looks like this
\begin{figure}[H]
  \includegraphics[scale=0.3,trim=25 25 25 25,clip]{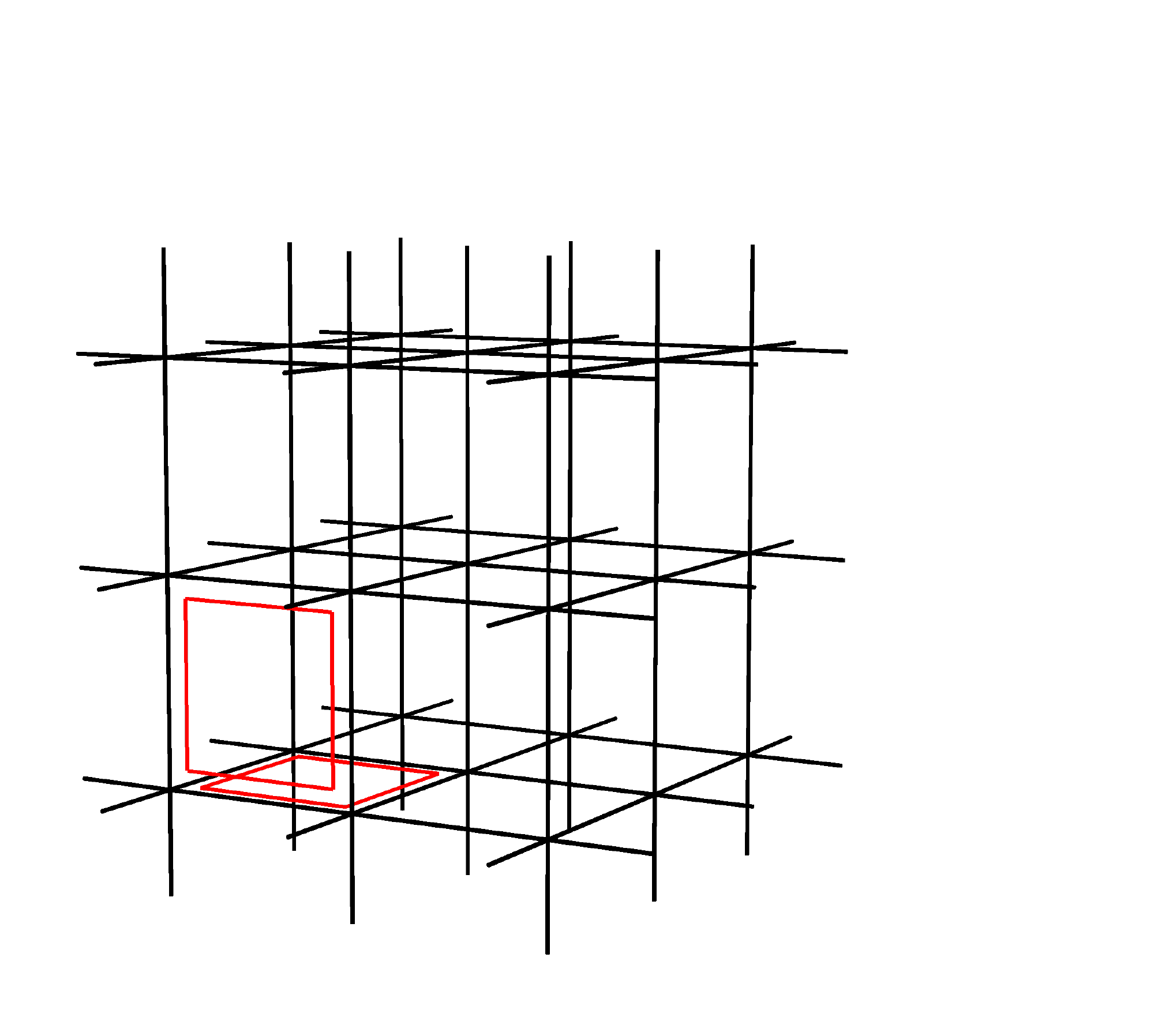}
    \caption{(Part of) a lattice $\Lambda$, d=3, with two plaquettes drawn.}\label{3DLattice}
    \end{figure} 
and the pumped-up lattice looks locally like this:
\begin{figure}[H]
  \includegraphics[scale=0.5,trim=25 25 25 25,clip]{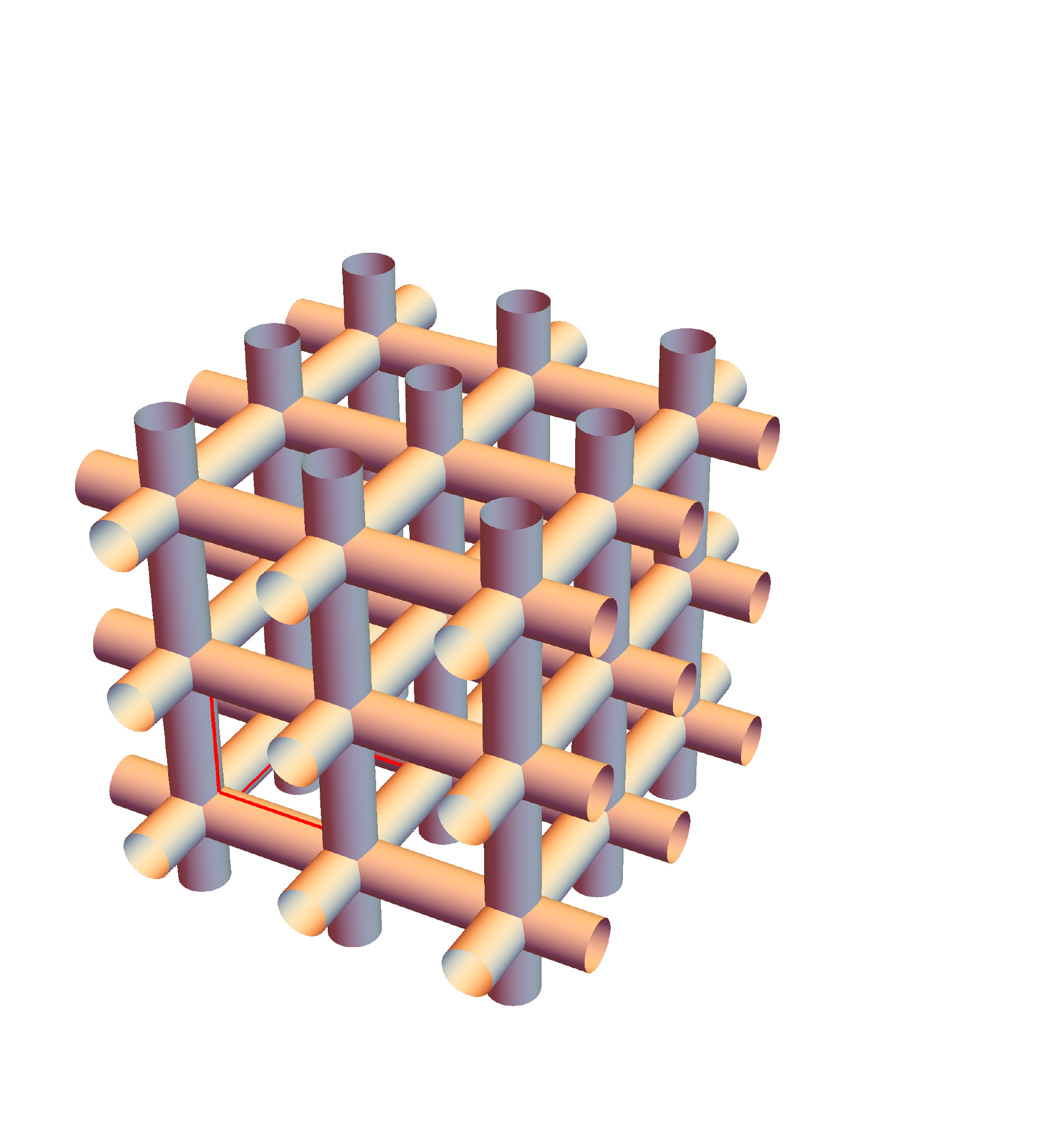}
    \caption{(Part of) the corresponding $M_\Lambda$, with (part of) two plaquette loops (lifted to $\bS_\Lambda$) shown.}\label{3DTubes}
    \end{figure}   

We need to choose a trinion decomposition $\tGamma$ that extends the set of plaquettes. To visualise this, let us think of $\Lambda$ as being embedded in the three dimensional torus $T^3_L=\bbR^3/L\bbZ^3$ in the obvious way, $M_\Lambda$ to be a regular neighbourhood and $\bS_\Lambda=\partial M_\Lambda$. If we take the ``dual lattice'' $\Gamma \subset T_M$ of the lattice $\Lambda$ (with vertices at the centres of the cubes defined by the edges of $\Lambda$ and edges joining vertices whenever two neighbouring cubes share a face), the closure $M_{\Gamma}$ of the complement of $M_\Lambda$ is a regular neighbourhood of $\Gamma$. (In this case $M_{\Gamma}$ is also a handlebody; together with $M_{\Lambda}$ it gives a Heegaard decomposition of $T^3_L$.)

The figure below shows a part of the handlebody $M_{\Gamma}$. Rather than taking the closure of the complement of $M_\Lambda$, we consider a smaller regular neighbourhood of $\Gamma$ for better visualisation. The six red circles correspond to the six plaquettes on six sides of the cube centred at the centre of the sphere (with six holes). They are among the loops in $\tsfP$. Adding the three blue loops on the surface of the sphere gives us remaining loops of the pants decomposition $\tGamma$ that are visible in this part of $\Sigma$.

\begin{figure}[H]
  \includegraphics[scale=0.8,trim=30 30 30 40]{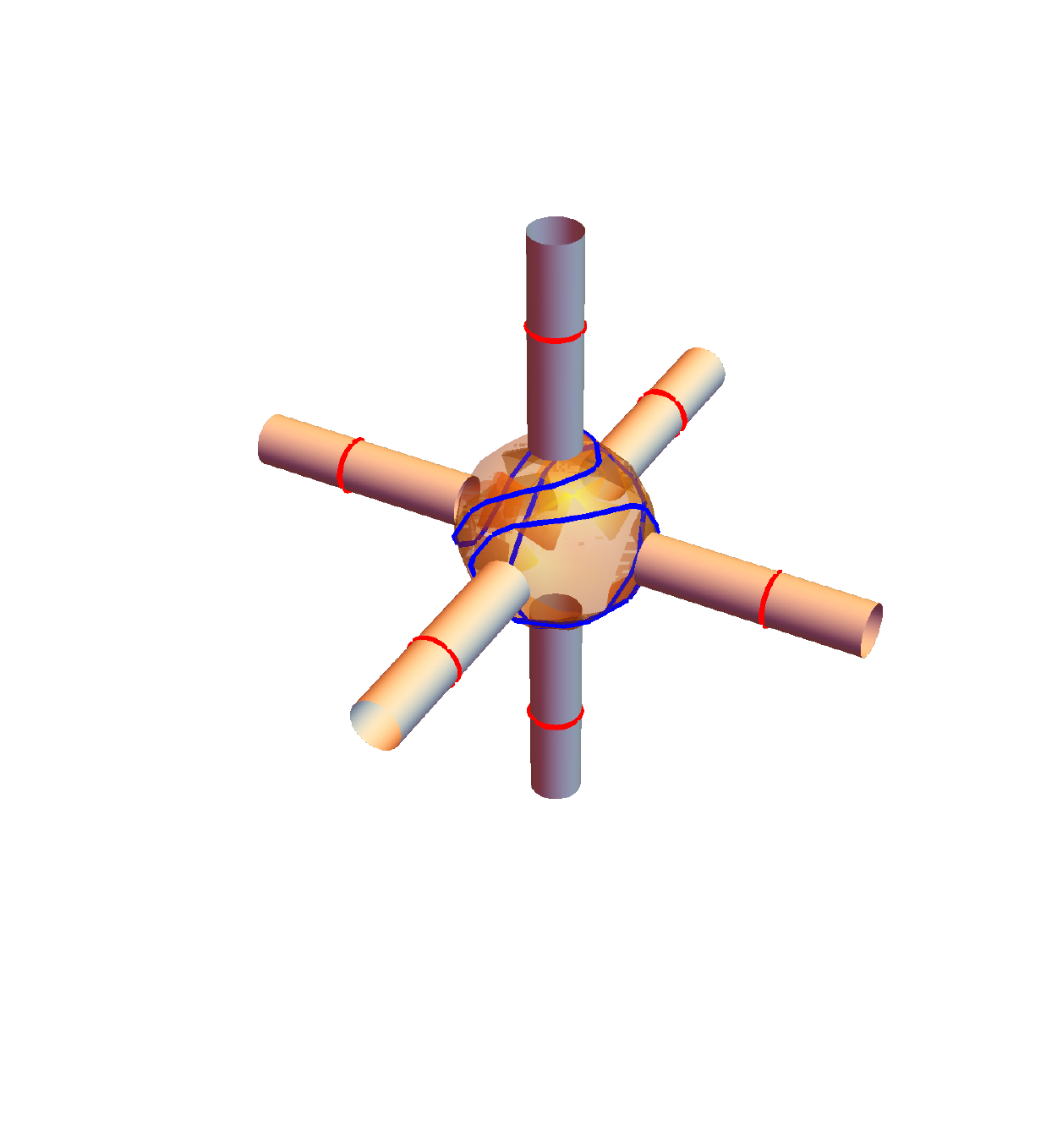}
    \caption{The trinion decomposition (locally at one vertex).}\label{3DTubesTrinionsLocal}
    \end{figure} 
 
The corresponding trinion graph is shown below. The red edges correspond to the red loops (the plaquettes of the lattice theory, lifted to the surface). the blue edges correspond to the loops added on each sphere-with-six-holes.

\begin{figure}[H]
  \includegraphics[scale=0.7,trim=120 200 30 30,clip]{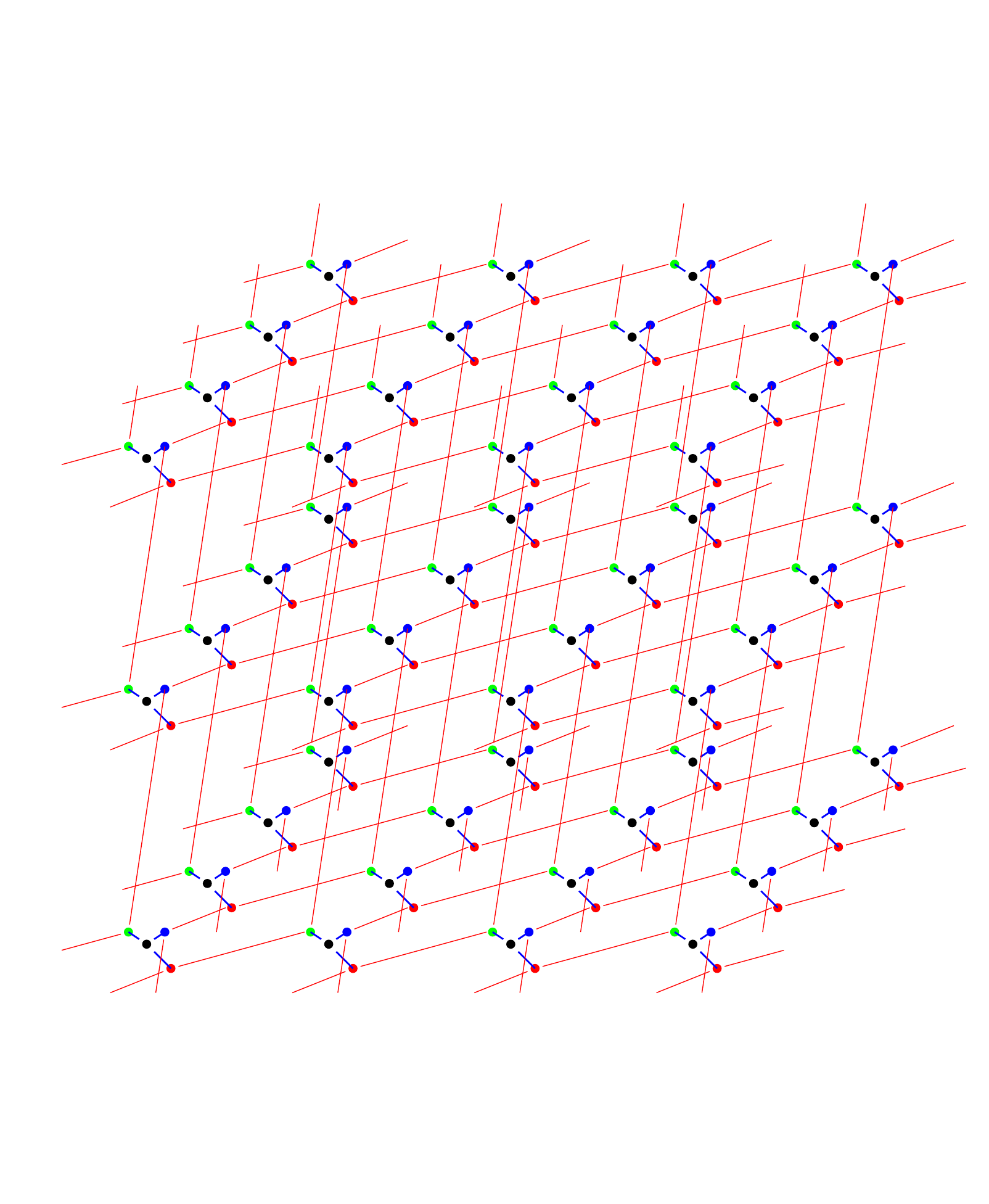}
    \caption{The trinion graph $\cT_\tGamma$.}\label{3DTrinionGraph}
    \end{figure} 

If the blue edges are contracted, we recover a copy of  the dual lattice $\Gamma$:

\begin{figure}[H]
  \includegraphics[scale=0.7,trim=110 200 30 30,clip]{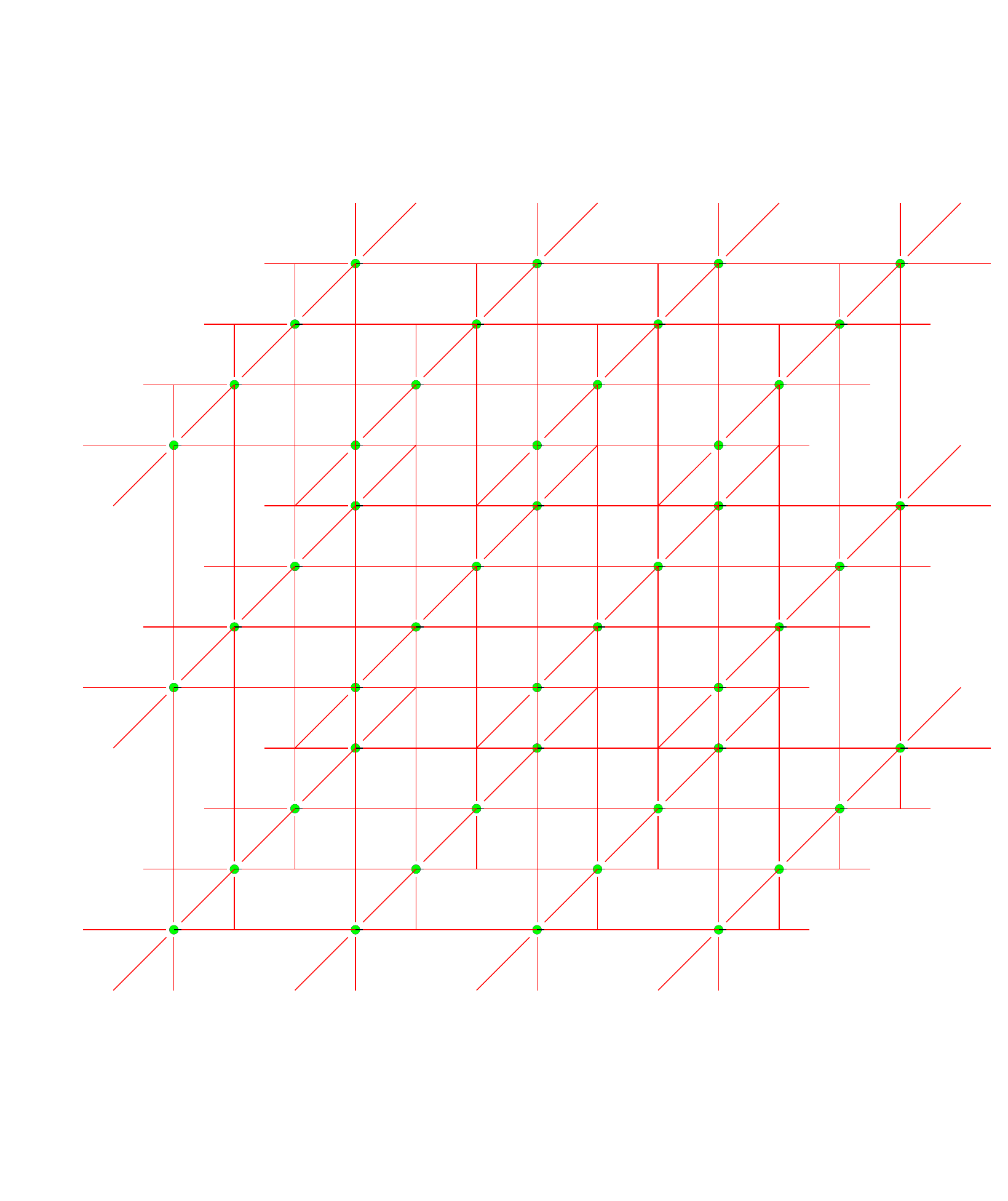}
    \caption{Contracting the blue links of $\cT_\tGamma$ yields $\Gamma$.}\label{3DTrinionGraphContracted}
    \end{figure} 

Before we move on, we recall some notation for easy reference.
\begin{enumerate}
\item $\Lambda$ is the graph (``lattice'') on which lattice gauge-fields live.
\item $M_\Lambda$ is the associated handlebody, and $\bS=\bS_\Lambda=\partial M_\Lambda$.
\item $\sfP$ is a set of loops (``plaquettes'') involved in the lattice path-integral weight.
\item $\tGamma$ is a set of loops in $\bS$ giving a pants decomposition, which we also denote by $\tGamma$. Under the retraction $\tgamma \mapsto \gamma \equiv r \circ \tgamma$, the subset $\tsfP \subset \tGamma$ maps bijectively to $\sfP$.
\end{enumerate}

\begin{remark} Consider the periodic lattice 
$$
\bbZ^3/L\bbZ^3 = \Lambda \subset T^3_L=\bbR^3/L\bbZ^3
$$
Recall that $\Gamma \subset T^3_L$ is the ``dual lattice'' and the pants decomposition $\tGamma$ of $\bS_\Lambda$ has as corresponding handlebody $M_{\Gamma}$, the closure of the complement of $M_\Lambda \subset T^3_L$.   Note that $M_{\Gamma}$ can be naturally identified with the handlebody $M_{\cT_\Gamma}$ associated to the trinion graph $\cT_\Gamma$.
The trinion graph $\cT_\tGamma$ is an alternate spine of $M_{\Gamma}$; contracting the edges of $\cT_\tGamma$ that do not correspond to plaquettes yields $\Gamma$.
\end{remark}

\subsection{``Nonabelian Gauss's Law'' II} Consider the six plaquettes that bound a unit cube in $\Lambda$. For any lattice gauge-field configuration $g$, the corresponding holonomies, with a common base-point and appropriate orientations, give six elements of $SU(2)$. In terms of the graph below, in which seven of the edges have been gauge-fixed to the identity, the holonomies of the six plaquettes, taken from the common base-point $O$, are: 
\[
\begin{split}
\mathtt{h}_1&=\mathtt{g}_2\\
\mathtt{h}_2&=\mathtt{g}_1^{-1}\\
\mathtt{h}_3&=\mathtt{g}_5^{-1}\\
\mathtt{h}'_1&=\mathtt{g}_4^{-1}\mathtt{g}_5\\
\mathtt{h}'_2&=\mathtt{g}_3\\
\mathtt{h}'_3&=\mathtt{g}_2^{-1}\mathtt{g}_3^{-1}\mathtt{g}_4\mathtt{g}_1\\
\end{split}
\]
We have $\mathtt{h}'_3\mathtt{h}_2\mathtt{h}_1'\mathtt{h}_3\mathtt{h}_2'\mathtt{h}_1=Id$, which we recognise as describing the fiundamental group of a sphere with six holes.

\begin{figure}[H]
  \includegraphics[scale=1.5,trim=30 10 50 30]{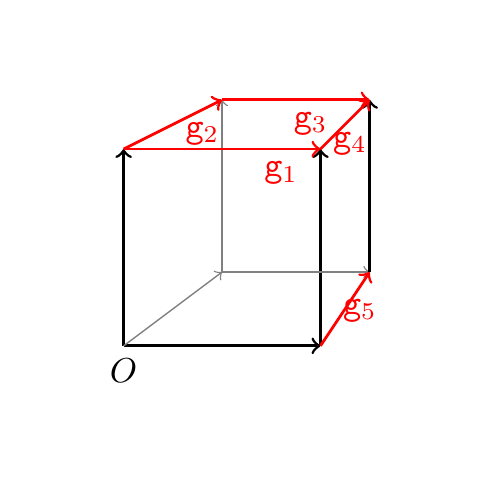}
    \caption{Edges in $\Lambda$ forming a cube. The black edges, which form a maximal tree in the cube, have been gauge-fixed to the identity.}
    \label{Cube}
    \end{figure} 
    
Given six characters $\chi_{l_i}, i=1,2,3,\ \chi'_{l_i}, i=1,2,3$ of $SU(2)$, we have
\[
\begin{split}
&\int \chi_{l'_3} (\mathtt{h}'_3(g))  \chi_{l_2} (\mathtt{h}_2(g)) \chi_{l'_1} (\mathtt{h}'_1(g)) \chi_{l_3} (\mathtt{h}_3(g)) \chi_{l'_2} (\mathtt{h}'_2(g)) \chi_{l_1} (\mathtt{h}_1(g)) \cD g\\
&= \int \chi_{l'_3} (\mathtt{g}_2^{-1}\mathtt{g}_3^{-1}\mathtt{g}_4\mathtt{g}_1)  \chi_{l_2} (\mathtt{g}_1^{-1}) \chi_{l'_1} (\mathtt{g}_4^{-1}\mathtt{g}_5) \chi_{l_3} (\mathtt{g}_5^{-1}) \chi_{l'_2} (\mathtt{g}_3) \chi_{l_1} (\mathtt{g}_2) d\mathtt{g}_1d\mathtt{g}_2 d\mathtt{g}_3 d\mathtt{g}_4 d\mathtt{g}_5\\
&=(\frac{1}{l'_3+1})^4 \delta_{l_1,l_2,l_3,l'_1.l'_2,l'_3} 
\end{split}
\]

The trinion graph, locally around the vertex representing the cube, looks like this:

\begin{figure}[H]
  \includegraphics[scale=.8,trim=140 250 60 60]{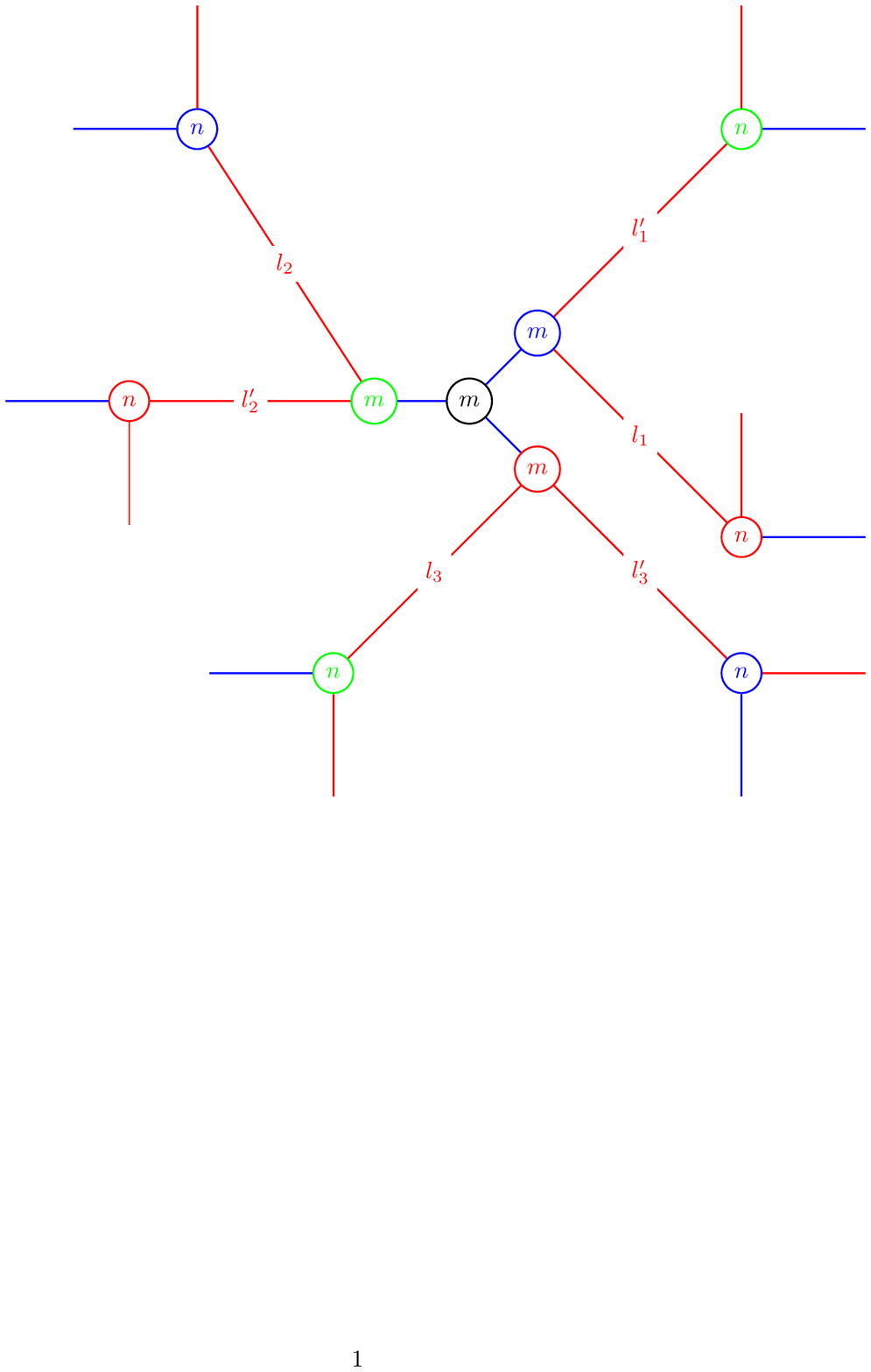}
    \caption{The trinion graph locally around the vertex representing the cube. Labels on the edges and vertices correspond to indices in Equation (\ref{GaussEquation}).}
    \label{TrinionCube}
    \end{figure} 

The integral over $\cM$ (corresponding to the above integral over $\bM$) gives
\begin{equation}\label{GaussEquation}
\begin{split}
&\frac{1}{\sum_{n=0}^\infty (\frac{\sqrt{2}}{(n+1) \pi})^{2(\sfb-1)}}\\
\times &\{
\sum_{n=0}^\infty (\frac{\sqrt{2}}{(n+1) \pi})^{2(\sfb-1)-4}  \sum_{m=0}^\infty (\frac{\sqrt{2}}{(m+1) \pi})^4\\
&\int \hchi_{l_1} (\pi t_1) \hchi_{l_2} (\pi t_2) \hchi_{l_3} (\pi t_3) \hchi_{l'_1} (\pi t'_1) \hchi_{l'_2} (\pi t'_2) \hchi_{l'_3} (\pi t'_3)\\ 
&\ \ \ \ \vu_n(t_1) \vu_m (t_1) \vu_n(t_2) \vu_m (t_2) \vu_n(t_3) \vu_m (t_3)\\ 
&\ \ \ \ \vu_n(t'_1) \vu_m (t'_1) \vu_n(t'_2) \vu_m (t'_2) \vu_n(t'_3) \vu_m (t'_3)\\  
&\ \ \ \  dt_1 dt_2 dt_3 dt'_1 dt'_2 dt'_3\}
\end{split}
\end{equation}

As $\sfb \to \infty$, this tends to
\[
\begin{split}
&(\frac{\pi}{\sqrt{2}})^4 \times \{\sum_{m=0}^\infty (\frac{\sqrt{2}}{(m+1) \pi})^4\\
&\int \hchi_{l_1} (\pi t_1) \hchi_{l_2} (\pi t_2) \hchi_{l_3} (\pi t_3) \hchi_{l'_1} (\pi t'_1) \hchi_{l'_2} (\pi t'_2) \hchi_{l'_3} (\pi t'_3)\\ 
&\ \ \ \ \ \hchi_m (\pi t_1)   \hchi_m (\pi t_2)   \hchi_m (\pi t_3)\hchi_m (\pi t'_1)   \hchi_m (\pi t'_2)   \hchi_m (\pi t'_3)\\  
&\ \ \ \  \tdt_1 \tdt_2 \tdt_3 \tdt'_1 \tdt'_2 \tdt'_3\}\\
&=(\frac{1}{l'_3+1})^4 \delta_{l_1,l_2,l_3,l'_1.l'_2,l'_3} 
\end{split}
\]

\subsection{Symplectic partition function with the Wilson action} Refer to the cubical lattice in Figure \ref{3DTrinionGraphContracted}. This is the ``contracted trinion graph'', $\Gamma$. The set of vertices of $\Gamma$ is denoted $\bbV_\Gamma$ and the set of (unoriented) edges $\bbV_\Gamma$.

Each vertex has an associated index $n_v$. Then the partition function of the symplectic theory is
\begin{equation}\label{partitionWilson}
\begin{split}
Z_\be = 
&\sum_{n_v=0}^\infty \prod_{v \in \bbV_\Gamma}  (\frac{\sqrt{2}}{\pi (n_v+1)})^4\\ &\times \prod_{e \in \bbE_\Gamma} \int_0^1 \vu_{n_{v_1(e)}}(t_e) \vu_{n_{v_2(e)}}(t_e) \exp{\{2 (\frac{2}{\be^2}) \cos{\pi t_e}\}} \ dt_e
\end{split}
\end{equation}
where $v_1(e)$ and $v_2(e)$ are the two vertices at either end of the (unoriented) edge $e$.

Recall that $\vu_l(t)=\sqrt{2} \hchi_l(\pi t) \sin{\pi t}$ and $2\cos{\pi t}=\hchi_1(\pi t)$. We will also use the relations:
\begin{enumerate}[label=(\alph*)]
\item $\hchi_p \hchi_q=\hchi_{|p-q|}+\hchi_{|p-q|+2}+\dots+\hchi_{p+q}\equiv \tsum_{r=|p-q|}^{|p+q|} \hchi_r$, where we have introduced the notation $\tsum$ for the sum that skips every alternate term in the usual sum $\sum$.
\item $\hchi_1^p=\sum_{k=0}^{\lfloor p/2 \rfloor} N_{p,k} \hchi_{p-2k}$, where 
$N_{p,k}={}^nC_{n-k}-{}^nC_{n-k-1}$ are multiplicities, and hence non-negative integers. 
\end{enumerate}

With all this, the contribution of an edge $e$ evaluates to:
\[
\begin{split}
&\int_0^1 \vu_{n_{v_1(e)}}(t_e) \vu_{n_{v_2(e)}}(t_e) \exp{\{2 (\frac{2}{\be^2}) \cos{\pi t_e}\}}\ dt_e\\
&=\tsum_{l_e=|n_{v_1(e)}-n_{v_2(e)}|}^{n_{v_1(e)}+n_{v_2(e)}}\ \ \sum_{m_e=0}^\infty \frac{(\frac{2}{\be^2})^{m_e}} {m_e!} \int_0^1  \hchi_{l_e}(\pi t_e) \hchi_1^{m_e} (\pi t_e) \tdt_e \\
&=\tsum_{l_e=|n_{v_1(e)}-n_{v_2(e)}|}^{n_{v_1(e)}+n_{v_2(e)}}\ \ \sum_{m_e=0}^\infty \frac{(\frac{2}{\be^2})^{m_e}} {m_e!} \int_0^1  \hchi_{l_e}(\pi t_e) \sum_{k_e=0}^{\lfloor m_e/2 \rfloor} N_{m_e,k_e} \hchi_{m_e-2k_e} (\pi t_e) \tdt_e \\
&= \tsum_{l_e=|n_{v_1(e)}-n_{v_2(e)}|}^{n_{v_1(e)}+n_{v_2(e)}} \sum_{m_e=0}^\infty \frac{(\frac{2}{\be^2})^{m_e}}{m_e!} \sum_{k_e=0}^{\lfloor m_e/2 \rfloor} N_{m_e,k_e} \delta_{l_e,m_e-2k_e}\\
&= \tsum_{l_e=|n_{v_1(e)}-n_{v_2(e)}|}^{n_{v_1(e)}+n_{v_2(e)}} \sum_{k_e=0}^\infty \frac{{(\frac{2}{\be^2})^{l_e+2k_e}}}{(l_e+2k_e)!}   N_{l_e+2k_e,k_e} \\
\end{split}
\]

With this, the partition function evaluates to:
\begin{equation}\label{partitionWilson}
\begin{split}
Z_\be = 
&\sum_{n_v=0}^\infty \prod_{v \in \bbV_\Gamma}  (\frac{\sqrt{2}}{\pi (n_v+1)})^4\\ &\times \prod_{e \in \bbE_\Gamma} \{ \tsum_{l_e=|n_{v_1(e)}-n_{v_2(e)}|}^{n_{v_1(e)}+n_{v_2(e)}} \sum_{k_e=0}^\infty \frac{ {(\frac{2}{\be^2})^{l_e+2k_e}} }{(l_e+2k_e)!}   N_{l_e+2k_e,k_e}\} \\
\end{split}
\end{equation}
(Recall that $v_1(e)$ and $v_2(e)$ are the two vertices at either end of the (unoriented) edge $e$.) This has the advantage that it is a sum of positive terms.

\subsection{The Migdal action} The weight
\[
\exp{\{{\frac{2}{\be^2} \sum_{plaquettes\ P}  W_P(g) }\}}=\prod_{plaquettes\ P} \exp{\{{\frac{2}{\be^2}  W_P(g) }\}}
\]
in the lattice path integral has the following qualitative behaviour: as $\be \to \infty$ the weight function tends to (the constant function) 1, so that the measure tends to $\cD g$, and as $\be \to 0$, the (normalised probability) measure concentrates around (gauge-equivalence class of) the gauge-field configuration $g_l=e$ for every oriented link $l$. The Migdal recipe  replaces the weight with
\[
\prod_{plaquettes\ P} \{\sum_{l_P=0}^\infty \exp{\{-\frac{l_P(l_P+1){\be'}^2}{2}\}}(l_P+1) \chi_{l_P}(H_P(g))\}
\]
which preserves the qualitative behaviour. The ``bare coupling constants'' $\be,\be'$ will have a complicated relationship depending on normalisation. In $d=2$, this replacement makes the lattice theory explicitly solvable on an arbitrary surface (\cite{Witten}). 

\subsection{Symplectic partition function with the Migdal action}  With the Migdal action the contribution of an edge $e$ is:
\[
\begin{split}
&\int_0^1 \vu_{n_{v_1(e)}}(t_e) \vu_{n_{v_2(e)}}(t_e) \{ \sum_{0}^\infty \exp{\{-\frac{l_e(l_e+1){\be'}^2}{2}\}}(l_e+1) \hchi_{l_e}(\pi t_e) \}  dt_e\\
&=\sum_{l_e} \exp{\{-\frac{l_e(l_e+1)}{2{\be'}^2}\}}(l_e+1)
\int_0^1 \hchi_{n_{v_1(e)}}(\pi t_e)  \hchi_{n_{v_2(e)}}(\pi t_e) \hchi_{l_e}(\pi t_e)  \tdt_e\\
&=\tsum_{l_e=|n_{v_1(e)}-n_{v_2(e)}|}^{n_{v_1(e)}+n_{v_2(e)}} \exp{\{-\frac{l_e(l_e+1){\be'}^2}{2}\}}(l_e+1)
\end{split}
\]

This gives the expression for the symplectic partition function with Migdal action:
\[
\begin{split}
Z'_{\be'} = 
&\sum_{\substack{n_v=0 \\ v \in  \bbV_\Gamma}}^\infty \ \ \prod_{v \in \bbV_\Gamma}  (\frac{\sqrt{2}}{\pi (n_v+1)})^4\\
&\times \prod_{e \in \bbE_\Gamma}  \{\tsum_{l_e=|n_{v_1(e)}-n_{v_2(e)}|} ^{n_{v_1(e)}+n_{v_2(e)}} \exp{\{-\frac{l_e(l_e+1){\be'}^2}{2}\}}(l_e+1)
\}
\end{split}
\]
This too is a sum of positive terms. 

\subsection{Plaquette-plaquette correlations} We can also write down expressions involving plaquette operators. For example, let $P_a,P_b$ be two plaquettes represented by edges $e_a,e_b$ in the trinion graph. Then, the symplectic plaquette-plaquette correlation with the Migdal action is
\[
\begin{split}
<W_{P_a}W_{P_b}>' &= \frac{1}{Z'_{\be'}} \sum_{\substack{n_v=0 \\ v \in  \bbV_\Gamma}}^\infty \ \ \prod_{v \in \bbV_\Gamma}  (\frac{\sqrt{2}}{\pi (n_v+1)})^4\\
&\times \prod_{\substack{e \in  \bbE_\Gamma\\ e \neq e_a,e_b}} \{\tsum_{l_e=|n_{v_1(e)}-n_{v_2(e)}|} ^{n_{v_1(e)}+n_{v_2(e)}} \exp{\{-\frac{l_e(l_e+1){\be'}^2}{2}\}}(l_e+1)
\}\\
\times \Bigg\{\{\tsum_{|n_{v_1({e_a})}-n_{v_2({e_a})}|}^{n_{v_1({e_a})}+n_{v_2({e_a})}} &\exp{\{-\frac{l_{e_a}(l_{e_a}+1){\be'}^2}{2}\}}(l_{e_a}+1)\}+\Delta_{n_{v_1(e_a)}+n_{v_2(e_a)}}-\Delta_{|n_{v_1(e_a)}-n_{v_2(e_a)}|}\Bigg\}
\\
\times \Bigg\{\{\tsum_{|n_{v_1({e_b})}-n_{v_2({e_b})}|}^{n_{v_1({e_b})}+n_{v_2({e_b})}} &\exp{\{-\frac{l_{e_b}(l_{e_b}+1){\be'}^2}{2}\}}(l_{e_b}+1)\}+\Delta_{n_{v_1(e_b)}+n_{v_2(e_b)}}-\Delta_{|n_{v_1(e_b)}-n_{v_2(e_b)}|}\Bigg\}
\end{split}
\]
We have used the computation:
\[
\begin{split}
&\int_0^1 \vu_{n_{v_1(e)}}(t_e) \vu_{n_{v_2(e)}}(t_e) \{ \sum_{l_e} \exp{\{-\frac{l_e(l_e+1){\be'}^2}{2}\}}(l_e+1) \hchi_{l_e}(\pi t_e) \}\cos{(\pi t_e)}  dt_e\\
&=\frac{1}{2} \sum_{l_e} \exp{\{-\frac{l_e(l_e+1)}{2{\be'}^2}\}}(l_e+1)
 \int_0^1 \hchi_{n_{v_1(e)}}(\pi t_e)  \hchi_{n_{v_2(e)}}(\pi t_e) \hchi_{l_e}(\pi t_e)  \hchi_1(\pi t_e) \  \tdt_e\\
 &=\frac{1}{2}\Bigg\{\tsum_{|n_{v_1(e)}-n_{v_2(e)}|-1}^{l_e \le n_{v_1(e)}+n_{v_2(e)}-1} \exp{\{-\frac{l_e(l_e+1){\be'}^2}{2}\}}(l_e+1)\\&+\tsum_{|n_{v_1(e)}-n_{v_2(e)}|+1}^{l_e \le n_{v_1(e)}+n_{v_2(e)}+1} \exp{\{-\frac{l_e(l_e+1){\be'}^2}{2}\}}(l_e+1)\Bigg\}\\
 &=\Bigg\{\tsum_{l_e=|n_{v_1(e)}-n_{v_2(e)}|}^{n_{v_1(e)}+n_{v_2(e)}} \exp{\{-\frac{l_e(l_e+1){\be'}^2}{2}\}}(l_e+1)\}+\Delta_{n_{v_1(e)}+n_{v_2(e)}}-\Delta_{|n_{v_1(e)}-n_{v_2(e)}|}\Bigg\}\\
\end{split}
\]
where
\[
\Delta_{l} \equiv \frac{1}{2}\{\exp{\{-\frac{(l+1)(l+2){\be'}^2}{2}\}}(l+2)-\exp{\{-\frac{l(l+1){\be'}^2}{2}\}}(l+1)\}
\]


\subsection{Symplectic analogue of a confinement criterion, with the Migdal action} Unlike the case $d=2$, it seems difficult to deal with the Wilson loop, since the natural pants decomposition extending plaquettes does not have large loops. We can instead consider, following \cite{TomboulisYaffe}, the analogue of the magnetic flux free energy. 

First, we allow our lattice to have different periodicities in the ``time'' and ``space'' directions -- respectively, $L_t,L_s$ -- so that the lattice has $N=L_s^2L_t$ vertices. We consider the modified symplectic partition function
\[
\begin{split}
\tZ'_{{\be'}} \equiv 
&\sum_{\substack{n_v=0 \\ v \in  \bbV_\Gamma}}^\infty \ \ \prod_{v \in \bbV_\Gamma}  (\frac{\sqrt{2}}{\pi (n_v+1)})^4\\
&\times \prod_{e \in \bbE_\Gamma \setminus \bbE_\ell}  \{\tsum_{l_e=|n_{v_1(e)}-n_{v_2(e)}|}^{n_{v_1(e)}+n_{v_2(e)}} \exp{\{-\frac{l_e(l_e+1){\be'}^2}{2}\}}(l_e+1)\}\\
&\times \prod_{e \in \bbE_\ell}  \{\tsum_{l_e=|n_{v_1(e)}-n_{v_2(e)}|}^{n_{v_1(e)}+n_{v_2(e)}} \exp{\{-\frac{l_e(l_e+1){\be'}^2}{2}\}}(l_e+1)\}^{-1}
\end{split}
\]
Here $\ell$ is a subgraph (with $L_t$ vertices and $L_t$ edges) running along the ``time'' axis. The symplectic analogue of the confinement criterion in \cite{TomboulisYaffe} is the following:
\[
\lim_{L_s \to \infty} \frac{\tZ'_{{\be'}}}{Z'_{{\be'}}}= 1
\]

Work is under way to analyse the above expressions, and also to cover the case $d=4$. As for the inclusion of fermions, it is here that a choice of complex structure on $\Sigma$ might be necessary.

\section{Further remarks}

\subsection{An alternate route to comparing integration on $\bM$ and $\cM$} Let $\Lambda$ be an arbitrary finite connected graph as at the beginning, $\bS$  the associated oriented surface, and $r:\bS \to \Lambda$ the retraction to the spine. One can choose an \emph{embedding} $\iota:\Lambda \hookrightarrow \bS$ which is a section for the retraction $r$. This gives a map $\iota^*:\cM \to \bM$, which is a retraction for the natural embedding $\bM \hookrightarrow \cM$, and one can hope to compare $\cD g$ and the push-forward of $d\mu_\cL$ by $\iota^*$. In fact, one can show, using Goldman flows, that the push-forward is independent of the choice of $\iota$. Further,
\begin{equation}\label{extend}
W_\tgamma=W_\gamma \circ \iota^*
\end{equation}
(Compare with Equation (\ref{restrict})). Unfortunately, given a family of plaquettes in $\Lambda$, it is in general impossible to deform their images by $\iota$ so that they become mutually disjoint loops in $\bS$.

\subsection{$U(1)$ lattice gauge theory}\label{U(1)Theory} In fact, the above strategy works for $U(1)$. The space of flat $U(1)$ connections on $\bS$ modulo $U(1)$ gauge transformations can be identified with the $2\sfb$-dimensional (real) torus $\cJ_{2\sfb}=H^1(\bS,\bbR)/H^1(\bS,\bbZ)$. If we choose an orientation on $\bS$, $\cJ$ is a symplectic manifold. The space of $U(1)$  lattice gauge fields on $\Lambda$, modulo lattice gauge transformations, on the other hand, is the $\sfb$-dimensional  subtorus $J_{\sfb}=H^1(J,\bbR)/H^1(J,\bbZ)$. If we embed $\Lambda$ in $\bS$, we get a retraction $\iota^*:\cJ_{2\sfb} \to J_{\sfb}$.  Normalise measures to total volume 1. In contrast to the case of $SU(2)$, the measure $\cD g$ coincides with the  push-forward of $d\mu_\cL$ by $\iota^*:\cJ_{2\sfb} \to J_{\sfb}$. This is easy to see -- both measures are invariant under translations on the torus $J_{\sfb}$. 

Consider briefly the case of $U(1)$ lattice theory in $d=3$. There are $3N$ plaquette loops.  Each plaquette inherits an orientation from the natural orientation in $\bbR^3$. Holonomies being gauge-invariant in this case, the set of holonomy $u_P$ for each (oriented) plaquette $P$, together with three ``global'' loops $u_x,u_y,u_z$ running along the three directions, together yield the analogue of the map $\bt_\bM$ -- in this case an injection
\[
\bu:J_{\sfb}=J_{2N+1}  \to U(1)^{3N+3}
\]
We have the following $N$ conditions on the image: for each unit cube $C$ in $\Lambda$:
\[
\prod_{P\ a\ face\ of\ C} u_P^{\pm 1}=1
\]
where the sign is chosen appropriately for each of he six faces of $C$. The other two conditions are that the products of holonomies for plaquettes in any $xy$-plane and in any $yz$-plane are both trivial. (This implies triviality of the product in any $zx$-plane.)

 \appendix

\section{Bernoulli Polynomials}

We follow the conventions in Hans Rademacher: \textit{Topics in analytic number theory}. For a short, clear account see: Omran Kouba: \textit{Lecture Notes, Bernoulli Polynomials and Applications}:\texttt{https://arxiv.org/abs/1309.7560}.

The Bernoulli polynomials $B_0,B_1,\dots,$ can be defined by the conditions
\begin{enumerate}
\item $B_0=1$
\item $B'_{q}(x)=qB_{q-1}(x)$
\item $\int_0^1 B_q(x) dx = 0$
\end{enumerate}
The Bernoulli number $b_k$ is defined by $b_k=B_k(0)$. 

It can be shown that

\begin{enumerate}
\item $B_1(x)=x-1/2$, $B_2(x)=x^2-x+1/6$, and more generally $B_q$ is a degree $q$ polynomial with rational coefficients.
\item $B_{q+1}(x+1)-B_{q+1}(x)=(q+1)x^{q}$, for $q \ge 0$.
\item On the interval $[0,1)$ we have:
$$
B_{2\sfb-1}(x)=(-1)^{\sfb} \frac{2(2\sfb-1)!}{(2\pi)^{2\sfb-1}} \sum_{k=1}^\infty \frac{\sin(2 \pi k x)}{k^{2\sfb-1}} \ \textup{for $\sfb>1$}
$$
and
$$
B_{2\sfb}(x)=(-1)^{\sfb+1}\frac{2(2\sfb)!}{(2\pi)^{2\sfb}} \sum_{k=1}^\infty \frac{\cos(2 \pi k x)}{k^{2\sfb}}\  \textup{for $\sfb\ge 1$} \ .
$$
\item As a consequence, on $[0,1)$ and hence everywhere
\begin{equation*}
\begin{split}
B_{2\sfb-1}(1-x)&=(-1)^\sfb \frac{2(2\sfb-1)!}{(2\pi)^{2\sfb-1}} \sum_{k=1}^\infty \frac{\sin(2 \pi k -2 \pi k x)}{k^{2\sfb-1}}\\
&=-(-1)^\sfb \frac{2(2\sfb-1)!}{(2\pi)^{2\sfb-1}} \sum_{k=1}^\infty \frac{\sin(2\pi k x)}{k^{2\sfb-1}}\\
&=-B_{2\sfb-1}(x)
\end{split}
\end{equation*}
and similarly
\begin{equation*}
B_{2\sfb}(1-x)=B_{2\sfb}(x)
\end{equation*}

\end{enumerate}

In the topology and algebraic geometry literature``Bernoulli Polynomials'' $P_q$ are used, which are related to the $B_k$ by
$$
q! P_q= B_q
$$

\section{A continuous version of the $SU(2)$ Verlinde Algebra}\label{trinionkernel}

If $t$ is the coordinate on $[0,1]$, we set $\tdt=2\sin^2{\pi t}\ dt$.

Let $\cP$ denote polytope in $\bbR^3$,  consisting of points $(t_1,t_2,t_3)$ subject the inequalities:
\begin{enumerate}
\item $0 \le t_i \le1,\ i=1,2,3$,
\item $t_1+t_2+t_3 \le 2$, and
\item $|t_2-t_3| \le t_1 \le t_2+t_3$.
\end{enumerate}
These conditions are invariant under permutation, being the conditions that $\pi t_1, \pi t_2, \pi t_3$ form the (angles subtended at the origin by) sides of a spherical triangle on a sphere of radius 1.

Fixing $0 \le t_3 \le 1$, we get the planar section $\cP_{t_3}$ of $\cP$ defined by the inequalities:
\begin{enumerate}
\item $0 \le t_i \le1,\ i=1,2$
\item $t_3 \le t_1+t_2 \le 2-t_3$
\item $|t_1-t_2| \le t_3$
\end{enumerate}

\begin{tikzpicture}
	\begin{axis}[ymin=0,ymax=1,enlargelimits=false,xlabel=$t_1$,ylabel=$t_2$]
	\addplot
	table {
		x y 
		0.2 0 
		0 0.2
		0.8 1
		1 0.8	
		0.2 0 	
	};
	 \node at (50,50) {$\cP_{0.2}$};
	\end{axis}
\end{tikzpicture}

Note that $area(\cP_{t_3})=2t_3(1-t_3)$, so that
\[
vol(\cP)=\int_0^1 2t_3(1-t_3) dt_3=t_3^2-\frac{2}{3}t_3^3 |^1_0=\frac{1}{3}
\]

Given an integrable real-valued function $f$ on $[0,1]$, we get a new function $N_{t_3}f$ by convolving with the indicator function of $\cP_{t_3}$:
\[
N_{t_3}f (t_2)=\int \mathbf{1}_{\cP_{t_3}} (t_2,t_1) f(t_1) dt_1\\
\]
Since the kernel is real, symmetric and bounded, this defines a self-adjoint Hilbert Schmidt operator $N_{t_3}:L^2([0,1]) \to L^2([0,1])$. By definition, given integrable functions $f_i,\ i=1,2,3$, we have
\[
\int_0^1 dt_1 f_3(t_3) (f_2, N_{t_3}f_1)_{L^2([0,1])}=\int \mathbf{1}_{\cP} (t_3,t_2,t_1) f_1(t_3)f_2(t_2)f(_1t_1) dt_3 dt_2 dt_3\
\]

\begin{proposition} A (normalised) eigenbasis  and corresponding eigenvalues  of $N_{t_3}$ are
\[
(\vu_n(t_1) \equiv \sqrt{2} \sin{(n+1) \pi t_1},\ \frac{2}{(n+1) \pi} \sin{(n+1) \pi t_3}),\ \ \ n=0,1, \dots
\]
\end{proposition}

\begin{remark} Note that the eigenbasis is independent of $t_3$. Note also that zero is an eigenvalue\footnote{and in that case it is of infinite multiplicity.} iff $t_3$ is rational.
\end{remark}

\begin{proof} The set $\{\vu_n|n=0,1,\dots,\}$ is an orthonormal basis for $L^2([0,1])$, where $\vu_n(t_1) \equiv \sqrt{2} \sin\{{(n+1)\pi t_1}\}$. We check by direct computation that these are eigenfunctions. 

Using the symmetry, we can suppose that $t_3 \le \frac{1}{2}$. We have three cases to consider:
\begin{enumerate}[wide, labelwidth=!, labelindent=0pt] 

\item If $t_2 \le t_3$, we have 
\[
\begin{split}
N_{t_3}\vu_n (t_2) &=\int_{t_3-t_2}^{t_3+t_2}  \sqrt{2} \sin{(n+1)\pi t_1} dt_1\\
&= \sqrt{2} \frac{1}{(n+1)\pi} [-\cos{(n+1)\pi t_1}]_{t_3-t_2}^{t_3+t_2}\\
&= \sqrt{2} \frac{1}{(n+1)\pi} (-\cos{(n+1)\pi (t_3+t_2)}+\cos{(n+1)\pi (t_3-t_2)})\\
&= \sqrt{2} \frac{2}{(n+1)\pi} \sin{(n+1)\pi t_3} \sin {n \pi t_2}\\
&= \frac{2}{(n+1) \pi} \sin{(n+1) \pi t_3}  \sqrt{2} \sin{(n+1)\pi t_2}\\
&= \frac{2}{(n+1)\pi} \sin{(n+1)\pi t_3} \ \vu_n(t_2)\\
\end{split}
\]
\item If $t_3 \le t_2 \le 1-t_3$, we have 
\[
\begin{split}
N_{t_3}\vu_n (t_2) &=\int_{t_2-t_3}^{t_3+t_2}  \sqrt{2} \sin{(n+1)\pi t_1} dt_1\\
&= \sqrt{2} \frac{1}{(n+1)\pi} [-\cos{(n+1)\pi t_1}]_{t_2-t_3}^{t_3+t_2}\\
&= \sqrt{2} \frac{1}{(n+1)\pi} (-\cos{(n+1)\pi (t_3+t_2)}+\cos{(n+1)\pi (t_2-t_3)})\\
&= \sqrt{2} \frac{2}{(n+1)\pi} \sin{(n+1)\pi t_3} \sin{(n+1)\pi t_2}\\
&= \frac{2}{(n+1)\pi} \sin{(n+1)\pi t_3}  \sqrt{2} \sin{(n+1)\pi t_2}\\
&= \frac{2}{(n+1)\pi} \sin{(n+1)\pi t_3} \ \vu_n(t_2)\\
\end{split}
\]
\item And finally, if $1-t_3 \le t_2$, we have 
\[
\begin{split}
N_{t_3}\vu_n (t_2) &=\int_{t_2-t_3}^{2-t_3-t_2}  \sqrt{2} \sin{(n+1)\pi t_1} dt_1\\
&= \sqrt{2} \frac{1}{(n+1)\pi} [-\cos{(n+1)\pi t_1}]_{t_2-t_3}^{2-t_3-t_2}\\
&= \sqrt{2} \frac{1}{(n+1)\pi} (-\cos{(n+1)\pi (2-t_3-t_2)}+\cos{(n+1)\pi (t_2-t_3)})\\
&= \sqrt{2} \frac{1}{(n+1)\pi} (-\cos{(n+1)\pi (t_3+t_2)}+\cos{(n+1)\pi (t_2-t_3)})\\
&= \sqrt{2} \frac{2}{(n+1)\pi} \sin{(n+1)\pi t_3} \sin{(n+1)\pi t_2}\\
&= \frac{2}{(n+1)\pi} \sin{(n+1)\pi t_3}  \sqrt{2} \sin{(n+1)\pi t_2}\\
&= \frac{2}{(n+1)\pi} \sin{(n+1)\pi t_3} \ \vu_n(t_2)\\
\end{split}
\]
\end{enumerate}

\end{proof}

\begin{remark} The previous proposition can be rewritten:
\[
N_{t_3}f (t_2) \equiv \int \mathbf{1}_{\cP_{t_3}} (t_2,t_1) f(t_1) dt_1=\sum_{n=0}^\infty 
\frac{2}{(n+1)\pi} \sin{(n+1)\pi t_3} (f,\vu_n)_{L^2} \vu_n
\]
This yields the ``Mercer expansion'':
\[
\mathbf{1}_{\cP_{t_3}}(t_2,t_1)= \sum_{n=0}^\infty 
\frac{\sqrt{2}}{(n+1)\pi} \sqrt{2} \sin{(n+1)\pi t_3} \ \sqrt{2} \sin{(n+1)\pi t_2} \sqrt{2} \sin{(n+1)\pi t_1}
\]
We will use the following version repeatedly:
\begin{equation}\label{Mercer}
\mathbf{1}_{\cP}(t_3,t_2,t_1)= \sum_{n=0}^\infty 
\frac{\sqrt{2}}{(n+1)\pi} \vu_n(t_1)\vu_n(t_2) \vu_n(t_3)
\end{equation}

The Hilbert-Schmidt norm of $N_{t_3}$ is the $L^2$ norm of its kernel. It is instructive to check this directly:
\[
\begin{split}
\sum_{n=0}^\infty \frac{4}{n^2 \pi^2} \sin^2{(n+1)\pi t_3} &= \sum_{n=0}^\infty \frac{4}{(n+1)^2 \pi^2}(1-\cos^2{(n+1)\pi t_3})\\
&=\frac{4}{\pi^2}\sum_{n=0}^\infty \frac{1}{(n+1)^2}-\frac{4}{\pi^2}\sum_{n=0}^\infty \frac{\cos^2{2(n+1)\pi \frac{t_3}{2}}}{(n+1)^2}\\
&=4\{\frac{1}{6}-B_2(t_3/2)\}\\
&=4\{-(t_3/2)^2+t_3/2\}=2t_3(1-t_3)=area(\cP_{t_3})
\end{split}
\]
\end{remark}

Likewise it is amusing to check that the $L^2$-norm-square of the indicator function $\mathbf{1}_{\cP}$ gives $vol(\cP)$:
\[
\frac{2}{\pi^2}\sum_{n=0}^\infty \frac{1}{(n+1)^2}=\frac{2}{\pi^2}\frac{\pi^2}{6}=\frac{1}{3}
\]

\begin{proposition} Given $u,v \in [0,1]$ the operator $N_uN_v$ is trace class and
\[
M(s,t) \equiv trace(N_sN_t)=
\begin{cases}
2t(1-s)\ \ if \ s \ge t\\
2s(1-t)\ \ if \ t \ge s\\
\end{cases}
\]
\end{proposition}

\begin{proof} The trace is given by the integral of the kernel along the diagonal. But $\int_0^1 \int_0^1 \mathbf{1}_{\cP_s} (u,v) \mathbf{1}_{\cP_t}(v,u) du dv$ is the area of the intersection $\cP_s \cap \cP_t$:

\begin{tikzpicture}
	\begin{axis}[ymin=0,ymax=1,enlargelimits=false,xlabel=$t_1$,ylabel=$t_2$]
	\addplot
	table {
		x y 
		0.2 0 
		0 0.2
		0.8 1
		1 0.8	
		0.2 0 	
	};
	 \node at (50,50) {$\cP_{0.2} \cap \cP_{0.3}$};
	 \addplot
	table {
		x y 
		0.3 0 
		0 0.3
		0.7 1
		1 0.7	
		0.3 0 	
	};
	\end{axis}
\end{tikzpicture}

Here is a more computational proof. Since the operators $N_s,N_t$ are simultaneously diagonalised,
\[
\begin{split}
trace(N_sN_t)&=\sum_{n=0}^\infty \frac{4}{(n+1)^2 \pi^2} \sin{(n+1)\pi s} \sin{(n+1)\pi t}\\
&=2\sum_{n=0}^\infty \frac{1}{(n+1)^2 \pi^2} 
\{
\cos{2(n+1)\pi(\frac{s-t}{2})}-\cos{2(n+1)\pi(\frac{s+t}{2})}
\}
\end{split}
\]
We have $0 \le \frac{s+t}{2} \le 1$; further, If $s \ge t$, we also have $0 \le \frac{s-t}{2} \le 1$. So the trigonometric sums can be identified with the Bernoulli polynomials,  and we get
\[
\begin{split}
trace(N_sN_t)&=2\{B_2(\frac{s-t}{2})-B_2(\frac{s+t}{2})\}\\
&=2\big\{\frac{(s-t)^2}{4}-\frac{s-t}{2}-\frac{(s+t)^2}{4}+\frac{s+t}{2}\big\}\\
&=2(t-st)=2t(1-s)
\end{split}
\]
\end{proof}

\begin{corollary} We have
\[
\int dt_3  \vu_{n_3}(t_3) N_{t_3}=\frac{\sqrt{2}}{(n_3+1) \pi} P_{n_3}
\]
where $P_{n}$ is the projection to the span of $\vu_{n}$.
\end{corollary}

\begin{proof} This is a direct computation.
\[
\begin{split}
\{[\int  \vu_{n_3}(t_3) N_{t_3}\ dt_3] \vu_{n_1}\}(t_2)&= \int \vu_{n_3}(t_3) 
\{N_{t_3} \vu_{n_1}\}(t_2) \  dt_3 \\
&=\int   \vu_{n_3}(t_3) 
\frac{2}{(n_1+1)\pi }\sin{(n_1+1) \pi t_3}\ \vu_{n_1}(t_2)\ dt_3 \\
&=\frac{\sqrt{2}}{(n_1+1) \pi } (\vu_{n_3},\vu_{n_1})_{L^2([0,1])} \vu_{n_1}(t_2)\\
&=\delta_{n_3,n_1}\frac{\sqrt{2}}{(n_1+1) \pi } \vu_{n_1}(t_2)
\end{split}
\]
\end{proof}

Let $M$ denote the integral operator with kernel $M(s,t)$.

\begin{proposition}  The (normalised) eigenfunctions  and corresponding eigenvalues  of $M$ are
\[
(\vu_n(t) \equiv \sqrt{2} \sin{(n+1) \pi t},\ \frac{2}{(n+1)^2 \pi^2} ),\ \ \ n=0,1, \dots
\]
\end{proposition}

\begin{proof} The first proof uses the definition:
\[
\begin{split}
M \vu_n (s)&=\int dt M(s,t) \vu_n(t)\\
&=\int dt \ trace(N_sN_t)\vu_n(t)\\
&=trace(N_s [\int dt\ \vu_n(t)N_t]))\\
&=\frac{\sqrt{2}}{(n+1) \pi}trace(N_s P_{n})\\
&=\frac{\sqrt{2}}{(n+1) \pi} \frac{2 \sin{(n+1) \pi s}}{(n+1) \pi}\\
&=\frac{2}{(n+1)^2 \pi^2} \vu_n (s)
\end{split}
\]

The second proof uses the expression for the  kernel $M(s,t)$:
\[
M(s,t) \equiv trace(N_sN_t)=
\begin{cases}
2t(1-s)\ \ if \ s \ge t\\
2s(1-t)\ \ if \ t \ge s\\
\end{cases}
\]

\[
\begin{split}
M \vu_n (s)&=\int  M(s,t) \vu_n(t)\ dt\\
&=\sqrt{2}\int_0^s \ 2t(1-s) \sin{(n+1)\pi  t}+\sqrt{2} \int_s^1 dt\  2s(1-t) \sin{(n+1)\pi  t} \ dt\\
&=-2s\sqrt{2}\int_0^1 \ t \sin{(n+1)\pi  t}+\sqrt{2}\int_0^s dt\ 2t \sin{(n+1)\pi  t}\\
&+\sqrt{2} \int_s^1 dt\  2s \sin{(n+1)\pi  t} \ dt\\
&=-2s\sqrt{2}\{-\frac{1}{(n+1) \pi}\cos{(n+1)\pi}\}+2\sqrt{2}\{\frac{1}{(n+1)^2\pi^2}\sin{(n+1)\pi s}\\
&-\frac{s}{n\pi}\cos{(n+1)\pi s}\}-\sqrt{2}\frac{2s}{(n+1)\pi} (\cos{(n+1)\pi }-\cos{(n+1)\pi s})\\
&=\frac{2}{(n+1)^2 \pi^2}\vu_n(s)
\end{split}
\]
\end{proof}

We have used the following:
\begin{lemma}
\[
\int_0^s t \sin{xt} \ dt=\frac{1}{x^2}\sin{xs}-\frac{s}{x}\cos{xs}
\]
\end{lemma}

\begin{proof} Set $I(x,s) \equiv -\int_0^s \cos{xt}\ dt=-\frac{1}{x}[\sin{xt}]^s_0=-\frac{1}{x}\sin{xs}$. Then
\[
\begin{split}
\int_0^s t \sin{xt} \ dt=\frac{d}{dx}I(x,s)&=\frac{1}{x^2}\sin{xs}-\frac{s}{x}\cos{xs}
\end{split}
\]
\end{proof}

\section{Volumes, etc.}\label{appendixvolumesetc}

Let $\cM_\sfb$ denote the moduli space of flat connections in genus $\sfb$. In this section, we assume that the symplectic form is normalised so that it represents a generator of the integral second cohomology group $H^2(\cM_\sfb,\bbZ))$. In this case, it is known that 
\[
vol(\cM_\sfb)=\frac{2}{2^{2(\sfb-1)}} vol(\cP_{\cT})
\]
where we can take $\cP_{\cT}$ to be the polytope corresponding to a(ny) trinion decomposition with trinion graph $\cT$.

We compute the volume in two ways. The first mimics the computation using the Verlinde formula, except that we use our continuous analogue. Note first that for any $l \ge 1$
\[
M^l \vu_n=(\frac{2}{(n+1)^2 \pi^2})^l \vu_n
\]
If we express $\Sigma$ as a cyclic union of $\sfb-1$ tori, and take the obvious periodic pants decomposition, the volume of the corresponding polytope is 
\[
\begin{split}
trace(M^{\sfb-1})&=\sum_{n=0}^\infty (\frac{2}{(n+1)^2\pi^2})^{\sfb-1}\\
&= (\frac{2}{\pi^2})^{\sfb-1} \sum_{n=0}^\infty \frac{1}{(n+1)^{2(\sfb-1)}}\\
&=\frac{2^{(\sfb-1)}}{\pi^{2(\sfb-1)}} \zeta((\sfb-1))\\
\end{split}
\]

Using the identity (\ref{Mercer}) easily gives the volume of polytope corresponding to \emph{any} pants decomposition:
\[
\begin{split}
\int_{\cP_\cT}  \mathbf{1} d\bt= \sum_{n=0}^\infty \{\frac{\sqrt{2}}{\pi (n+1)}\}^{2(\sfb-1)}=\frac{2^{(\sfb-1)}}{\pi^{2(\sfb-1)}} \zeta(2(\sfb-1))\\
\end{split}
\]

Either way, we get
\[
vol(\cM_\sfb)=\frac{2}{(2\pi^2)^{(\sfb-1)}} \zeta(2(\sfb-1))
\]

Let us now find the distribution of one elementary Wilson loop as a function of genus $\sfb$. Given any function $f$ of $t$,
\[
\begin{split}
\int_{\cM_\sfb} f(t) d\mu_\cL &= \frac{2}{2^{2\sfb-2}} \int dt dt_1 \dots dt_{\sfb-2}\  f(t) \vu_n(t) \vu_n(t) M(t,t_1) M(t_1,t_2) \dots M(t_{\sfb-2},t)\\
&=2(\frac{2}{\pi^2})^{\sfb-1} \sum_{n=0}^\infty \frac{1}{(n+1)^{2(\sfb-1)}} \int dt\  f(t) \vu_n(t) \vu_n(t)
\end{split}
\]
So the normalised integral has the expression
\[
\begin{split}
\frac{\int_{\cM_\sfb} f(t) d\mu_\cL}{vol(\cM_\sfb)}&=\frac{\int dt\  f(t) \times 2 \sum_{n=0}^\infty \frac{\sin^2{(n+1) \pi  t}}{n^{2(\sfb-1)}}}{ \sum_{n=0}^\infty \frac{1}{(n+1)^{2(\sfb-1)}} }
\end{split}
\]
We rewrite
\[
\begin{split}
\frac{2 \sum_{n=0}^\infty \frac{\sin^2{(n+1) \pi l t}}{n^{2(\sfb-1)}}}{ \sum_{n=0}^\infty \frac{1}{(n+1)^{2(\sfb-1)}} }&=1-\frac{\sum_{n=0}^\infty \frac{\cos{2(n+1) \pi t}}{(n+1)^{2(\sfb-1)}}}{\sum_{n=0}^\infty \frac{1}{(n+1)^{2(\sfb-1)}}}=1-\frac{B_{2\sfb-2}(t)}{B_{2\sfb-2}(0)}
\end{split}
\]
By Corollary 5.2(i) of [Kouba],
$$
\lim_{\sfb \to \infty} (-1)^\sfb \frac{(2\pi)^{2\sfb-2}}{2(2\sfb-2)!} B_{2\sfb-2} (t)= \cos(2\pi t)
$$
uniformly on compacts. This gives
\[
\begin{split}
\lim_{\sfb \to \infty} \frac{\int_{\cM_\sfb} f(t) d\mu_\cL}{vol(\cM_\sfb)}&= \int_0^1 f(t)\ 2\sin^2{\pi t}\ dt=\int_0^1 f(t) \tdt
\end{split}
\]

 \end{document}